\documentclass[
letterpaper,
aps,
accepted=2026-08-25
]{quantumarticle}

\pdfoutput=1

\usepackage{graphicx} 
\usepackage[colorlinks=true, linkcolor=blue, citecolor=blue,  filecolor=blue, urlcolor=blue]{hyperref}
\usepackage{natbib}
\usepackage{amsmath}
\usepackage{amsthm}
\usepackage{amsfonts}
\usepackage{subfigure}
\usepackage{dsfont}
\usepackage{eufrak}
\usepackage{tikz}
\usepackage{braket}
\usepackage{listings}
\usepackage[margin=0.8in]{geometry}
\usepackage{comment}
\usepackage{appendix}
\usepackage[normalem]{ulem}
\usepackage{makecell}

\newtheorem{theorem}{Theorem}
\newtheorem{lemma}[theorem]{Lemma}
\newtheorem*{lemma*}{Lemma}

\newtheorem{claim}[theorem]{Claim}
\newtheorem*{claim*}{Claim}

\theoremstyle{definition}
\newtheorem{definition}[theorem]{Definition}

\newcommand{\dev}[2]{\frac{\text{d} #1}{\text{d} #2}}

\newcommand{\iprod}[2]{\langle #1 | #2 \rangle}
\newcommand{\n}{\notag \\}

\newcommand{\tlx}{\asymp}
\newcommand{\tli}{\scriptstyle{)(}}

\newcommand{\bbbbbbbb}{\!\!\!\!\!\!\!\!}

 \usepackage{xcolor}
\usepackage{amsmath}

\begin{document}

\def\udem{D\'epartement de Physique, Universit\'e de Montr\'eal, Montr\'eal, QC, Canada H3C 3J7}
\def\udemB{Centre de Recherches Math\'ematiques, Universit\'e de Montr\'eal, Montr\'eal, QC, Canada H3C 3J7}
\def\udemC{Institut Courtois, Universit\'e de Montr\'eal, Montr\'eal, QC, Canada H2V 0B3}

\title{Spatial structure of multipartite entanglement at measurement induced phase transitions}
\author{James Allen}
\affiliation{\udem}
\author{William Witczak-Krempa}
\affiliation{\udem}
\affiliation{\udemB}
\affiliation{\udemC}
\date{September 4, 2026} 

\begin{abstract}
    We study multiparty entanglement near measurement induced phase transitions (MIPTs), which arise in ensembles of local quantum circuits built with unitaries and measurements. In contrast to equilibrium quantum critical transitions, where entanglement is short-ranged, MIPTs possess long-range k-party genuine multiparty entanglement (GME) characterized by an infinite hierarchy of entanglement exponents for k$\geq$2. First, we represent the average spread of entanglement with ``entanglement clusters'', and use them  to conjecture general exponent relations: 1) classical dominance, 2) monotonicity, 3) subadditivity. We then introduce measure-weighted graphs to construct such clusters in general circuits.
    Second, we obtain the exact entanglement exponents for a 1d MIPT in a  measurement-only circuit that maps to percolation by exploiting non-unitary conformal field theory. The exponents, which we numerically verify, obey the inequalities. 
    We also extend the construction to a 2d MIPT that maps to classical 3d percolation, and numerically find the first entanglement exponents. Our results provide a firm ground to understand the multiparty entanglement of MIPTs, and more general ensembles of quantum circuits.
\end{abstract}

\maketitle

\section{Introduction}

Genuine multipartite entanglement (GME) occurs in a multipartite system where all parties contribute to the quantum correlations. It is a very powerful resource for many quantum algorithms~\cite{Raussendorf2001,Sorensen2001,Scott2004,Briegel2009}, but is generally more elusive than its bipartite counterpart.
It turns out that GME is generally fragile in subregions of quantum many-body systems (such as spin chains) when one spatially separates the parties~\cite{Osterloh_2002,Osborne,Javanmard2018,Parez2024b}, or under time-evolution following a quantum quench~\cite{Parez2024}. This ``fate of entanglement'' even holds in quantum critical systems in equilibrium~\cite{Parez2024b}. 
However, GME is also known~\cite{Sang2020,Avakian2024} to persist at longer range (i.e. with a power-law decay) in certain \textit{non}-unitary systems. Important examples of these systems are random quantum circuit (RQC) ensembles in the presence of measurements~\cite{Nahum2016,Bouland2019,Bouland2019a,Potter2021,Fisher2023}. These circuits in general have many interesting features, from their potential as error-correcting codes~\cite{Brandao2016} to their use in demonstrating quantum supremacy~\cite{Boixo2018,Napp2022}, but one of the most important properties is the presence of a phase transition when the measurement rate reaches a certain threshold~\cite{Li2018,Chan2019,Li2019,Skinner2019,Vasseur2019,Zabalo2020,Bera2020,Ware2023arxiv}. The critical point of the measurement induced phase transition (MIPT) provides the long-range GME.

At criticality, the non-unitary system will have $k$-party entanglement with power-law scaling at some exponent $\alpha_k$. That is, given a function $E_k(A_1, ..., A_k)$ that measures ensemble-averaged $k$-party GME over the reduced density matrix (RDM) of subregions $A_1, ..., A_k$ in a larger system, if all distances between subregions are scaled by a factor $\lambda\gg 1$ (without changing the sizes of the subregions themselves), we expect $E_k$ to asymptotically scale as: 
\begin{gather}\label{eq:gme_scaling}
    E_k \sim \frac{1}{\lambda^{\alpha_k}} 
\end{gather}
The key to understanding GME in these ensembles lies in the behaviour of the entanglement exponents $\alpha_k$. 

In this paper, we will first claim that the entanglement exponents $\alpha_k$ follow some general guidelines. Firstly, we consider an equivalent exponent for the $k$-party mutual information, $\alpha_k^{\text{MI}}$, which measures both quantum and classical multiparty correlations. If such an exponent exists, it must always be at most $\alpha_k$. Secondly, $\alpha_k$ is monotonically increasing with $k$. Our last and strongest claim is that $\alpha_k$ is subadditive in $k$ - which guarantees, for example, that these systems will not have an abnormally suppressed higher-party GME. We present our arguments for these guidelines in Section~\ref{section:subadditivity}, using a heuristic picture of entanglement in terms of cluster connection conditions.

We will test our claims with numerical simulations of a class of RQC ensembles known as measurement-only circuits (MOC). Most RQC ensembles are difficult to simulate classically~\cite{Movassagh2019arxiv,Dalzell2022,Napp2022,Chen2024conference}, unless there exists a comprehensive set of stabilizers that can remove the exponential cost in system size~\cite{Gottesman1998,Bravyi2016,Audenaert2005}, and it is difficult to determine the $k$-party GME of a generic mixed state in a way that is both efficient and accurate~\cite{Guhne2009,Guhne2010,Ma2011b,Lyu2024,Audenaert2005}. The MOC of the projective Ising model~\cite{Nahum2020,Lang2020,Lavasani2021,Ippoliti2021}, consisting of nothing but single-site $X$ measurements and nearest-neighbor $ZZ$ measurements, is simple enough that all of these difficulties are mitigated, while still providing an interesting entanglement structure typical of critical non-unitary systems, due to its mapping to 2D bond percolation.

\setlength{\tabcolsep}{5pt}
\begin{table}[]
    \centering
    \begin{tabular}{|c| c| c c c c |} \hline
        System & \thead{Theory \\ $\alpha_k$} & $\alpha_2$ & $\alpha_3$ & $\alpha_4$ & $\alpha_5$ \\ \hline 
        1+1D & $2k$ & 4.01(5) & 6.2(3) & 8.2(4) & 10.1(8) \\ \hline 
        2+1D & $3k$? & 6.2(3)& 9.3(8) & 10-15 & - \\ \hline 
        \hline 
        System & $\alpha_k^{\text{MI}}$ & $\alpha_2^{\text{MI}}$ & $\alpha_3^{\text{MI}}$ & $\alpha_4^{\text{MI}}$ & $\alpha_5^{\text{MI}}$ \\ \hline 
        1+1D & $k/3$ & 0.67 & 0.99 & 1.31 & 1.64 \\ \hline 
        2+1D & $0.97k$ & 2.05(10) & 2.96(4) & 3.92(4) & - \\ \hline 
    \end{tabular}
    \caption{Measured values of the $k$-party entanglement power law exponent $\alpha_k$ and the mutual information power law exponent $\alpha_k^{\text{MI}}$, for 1+1D and 2+1D systems, compared to theoretical predictions (shown in the second column). The $\alpha_k=3k$ prediction for 2+1D systems is a hypothetical extension of the 1+1D theory. The theoretical value of $\alpha_k^{\text{MI}}$ for the 2+1D theory is derived from the anomalous scaling exponent of Ref.~\onlinecite{Baek2010}.}
    \label{tab:slopes}
\end{table}

Ref.~\onlinecite{Sang2020} previously measured bipartite GME in this system, obtaining $\alpha_2 = 4$ both numerically and theoretically through the CFT of the percolation model. We extend the theoretical result in Section~\ref{section:theory} to show that $\alpha_k = 2k$ for general $k$, and numerically verify this in Section~\ref{section:numerics} up to $k=5$, also showing classical dominance and monotonicity of entanglement exponents in this regime. We cover the 2+1D version of this circuit in Section~\ref{section:3d}, which maps to the analytically unsolved 3D bond percolation model. Our measured and predicted values of the entanglement scaling exponent $\alpha_k$ and the mutual information scaling exponent $\alpha_k^{\text{MI}}$ can be found in Table~\ref{tab:slopes}. Both circuits easily satisfy classical dominance and monotonicity of $\alpha_k$, and the 1+1D circuit saturates subadditivity, with the numerical evidence of the 2+1D circuit suggesting exponents that are near additive as well.

\section{Overview of $k$-party Entanglement} \label{section:subadditivity}

We begin with a review of GME. A bipartite state is entangled if there exist correlations between the parties that are not classical. Unentangled states are classical mixtures of product states, and belong to the set of separable states $\rho_{A_1 A_2} = \sum_j p_j \rho_{A_1,j} \otimes \rho_{A_2,j}$, 
where $A_1|A_2$ are the two parties in the system. A state over these parties is still correlated through the mixing variable $j$, and therefore can have nonzero mutual information, but cannot be used as a quantum resource.

A general density matrix over multiple parties has GME if it cannot be written as a biseparable state $\rho = \sum_{m|\overline{m},j} p_{m,j} \rho_{m,j} \otimes \rho_{\overline{m},j}$, 
where the sum is now over partitions of all parties $A_1,...,A_k$ into two subsets $m |\overline{m}$. The set of biseparable density matrices includes states that may be bipartite entangled along some parties but not all of them, as well as classical mixtures of such states. Like in the bipartite case, GME measures differ from the $k$-party mutual information, which measures all the information shared by the parties, classical or quantum.

In order to define $E_k$, we need to choose a specific function to measure the GME of a particular state. There is no single ``perfect" measure of GME - instead, there is a diverse set of different measures such as the genuine multiparty negativity~\cite{Jungnitsch2011} or the geometric entanglement~\cite{Wei2003} which satisfy some desirable properties but not others. A typical choice for the entanglement measure $E_k(A_1, ..., A_k)$ is a monotone - it must be zero for all biseparable states, and nonincreasing under operations local to a single party or classical communication (LOCC). We will assume that $E_k$ is a monotone for the rest of this paper.

\begin{figure}[h]
    \centering
    \begin{tikzpicture}
        \begin{scope}
            \node[anchor=south west,inner sep=0] (image_a) at (0,0)
            {\includegraphics[width=0.5\columnwidth]{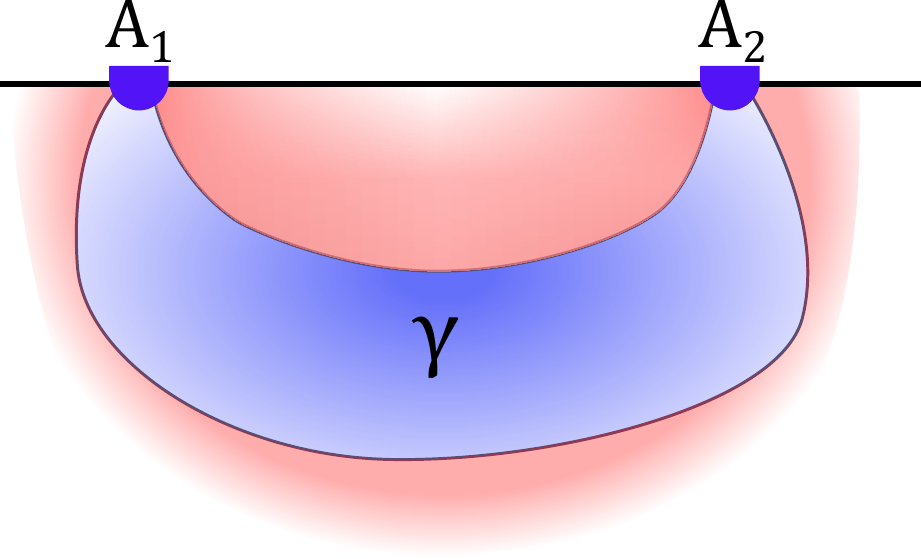}};
            \draw [->] (0,0) -- (1,0) node[right]{\footnotesize{$x$}};
            \draw [->] (0,0) -- (0,1) node[right]{\footnotesize{$t$}};
        \end{scope}
    \end{tikzpicture}
    \vspace*{-0.2cm}
    \caption{A pictorial representation of $2$-party entanglement between subregions $A_1, A_2$ in a 1+1-dimensional system. $x$ indicates the spatial direction while $t$ is the time direction. In order for there to be large $k$-party entanglement through a cluster $\gamma$, there should be high entanglement regions \textit{(blue)} through the bulk of $\gamma$ and low entanglement regions \textit{(red)} between $\gamma$ and the complement of $A=A_1A_2$. }
    \label{fig:entanglement_region_picture_2}
\end{figure}

\subsection{Cluster Picture of Entanglement}
We can construct a heuristic picture to quantify $k$-party entanglement in RQCs, and their ensembles, through entanglement clusters. Similar clusters were studied in the context of monitored Clifford circuits in Ref.~\onlinecite{Lunt2021} through connected components of a graph state - here, we consider a general monitored RQC and extend the clusters into the circuit history, or ``bulk'', of the system.

Let us consider an RQC that prepares final states over a probability distribution. We want to know whether parties $A_1\cdots A_k$ in the final state share GME. In order to maximise GME, the circuit needs to generate strong entanglement within $A = \bigcup_{j} A_j$, and to limit entanglement between $A$ and its complement $\bar{A}$. The latter follows from the monogamous nature of entanglement.
The starting point of our construction is the ``minimal spanning graph''~\cite{Avakian2024}: the shortest subgraph connecting the $k$ parties, where dis-entangling operations (like single-site measurements) correspond to removing an edge between qubits in the graph representation of the quantum circuit. The strength of GME within $A$ is limited by the monogamous property of entanglement: entanglement between $A$ and $\bar{A}$ should be low. This amounts to limiting ``parasitic'' edges between $A$ and $\bar{A}$. This construction suggests the following generalisation: the entanglement can be described by a \emph{cluster} $\gamma$ that connects the subregions $A_1 \cdots A_k$ through the bulk, with high entanglement inside the cluster, but low entanglement on the boundary of the cluster. An illustration of such a bipartite cluster is given in Fig.~\ref{fig:entanglement_region_picture_2}.
As a physical analogy, we can assign a free energy cost to high entanglement regions, which then acts as an effective tension on $\gamma$. On the other hand, the free energy cost of low entanglement regions acts as a repulsion between $\gamma$ and the complement of $A$ on the surface. 

In the following subsection, we will present several conjectures on the nature of the scaling exponents $\alpha_k$ in systems with long-range GME, with arguments in their favor based on the cluster picture. In Section~\ref{section:numerics}, we introduce an explicit construction, the \emph{entanglement-weighted graph}, that allows us to identify entanglement clusters in numerical simulations of quantum circuits. 

\subsection{Entanglement exponent relations}\label{section:entanglement_exponent_conjectures}
At a critical point, an ensemble of circuits will possess $k$-party GME that scales as (\ref{eq:gme_scaling}). Likewise, $k$-party mutual information is expected to scale as $\lambda^{-\alpha_k^{\text{MI}}}$, where we introduced the mutual information exponent $\alpha_k^{\text{MI}}$. 
We make the following conjectures for the GME entanglement exponents:
\begin{align}
    \alpha_k &\geq \alpha_k^{\text{MI}} \label{eq:classical_dominance}\\
    \alpha_{k+1} &\geq \alpha_k \label{eq:monotonicity_of_entanglement}\\
    \alpha_k + \alpha_\ell &\geq \alpha_{k+\ell} \label{eq:subadditivity_of_entanglement}
\end{align}
which represent dominance of classical correlations over entanglement, monotonicity of entanglement, and subadditivity of entanglement exponents, respectively.

The arguments for classical dominance and monotonicity are relatively simple. Because $k$-party mutual information does not suffer from a monogamy condition, the requirements for a cluster $\gamma$ yielding $k$-party MI are relaxed. In particular, the low-entanglement condition between $\gamma$ and the surface complement does not have to be as strong. Any entanglement between $\gamma$ and the surface complement becomes a correlation with a random variable when the surface complement is traced out, which harms quantum correlations between the subregions connected by $\gamma$ but not classical correlations. Therefore, instances of $k$-party MI are more abundant than $k$-party entanglement.

For monotonicity, we start with a $k+1$-party entanglement condition on subregions $A_1,...,A_{k+1}$. By bringing the last two subregions $A_k, A_{k+1}$ together, we must increase $k+1$-party entanglement, as distances between subregions have only been reduced (Fig.~\ref{fig:entanglement_monotonicity}). Therefore, by combining the two subregions together into a larger, new subregion $A_k'$, we must have 
\begin{align}
    E_{k+1}(A_1, ..., A_{k+1}) &\leq E_k(A_1, ..., A_{k-1}, A_k')\n
    &\leq E_k(A_1', ..., A_k')
\end{align}
where $A_1', ..., A_{k-1}'$ have the same centers as $A_1, ..., A_{k-1}$, but with width scaled by $|A_k'|/|A_k|$. If the subregions are equal width, this width scale is exactly $2$ - therefore, by invoking scale invariance, this is equivalent to a $k$-party cluster with subregions of the same width as $A_k$, but separated by half the distance, and therefore has entanglement equal to $2^{\alpha_k}E_k(A_1,...,A_k)$. Hence, the $k$-party entanglement between subregions $A_1, ..., A_k$ upper bounds the $k+1$-party entanglement, up to a constant.

\begin{figure}[h]
    \centering
    \begin{tikzpicture}
        \begin{scope}
            \node[anchor=north west,inner sep=0] (image_a) at (0,0)
            {\includegraphics[width=0.95\columnwidth]{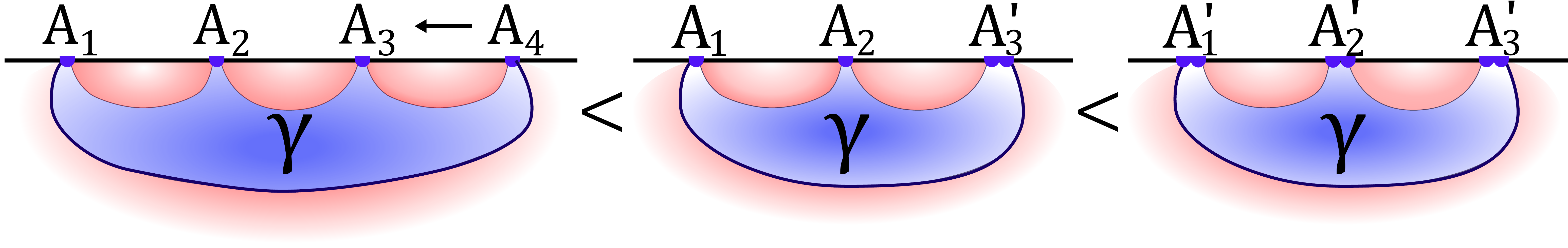}};
        \end{scope}
    \end{tikzpicture}
    \vspace*{-0.2cm}
    \caption{Entanglement monotonicity in the $k=3$ case - a 4-party cluster of expected entanglement $E_4(A_1,A_2,A_3,A_4)$ can have its two last subregions $A_3, A_4$ merged, then each other region doubled, to obtain a 3-party cluster with double the subregion width. In each step, the expected entanglement increases. }
    \label{fig:entanglement_monotonicity}
\end{figure}

To argue for subadditivity, let us consider the case $k=\ell=2$. We will start with independent 2-party cluster conditions $\gamma_A, \gamma_B, \gamma_C$ that entangle $A_1\leftrightarrow A_2$, $B_1 \leftrightarrow B_2$, and $[A]\leftrightarrow [B]$ respectively (where $[A]$ is the entire space from the middle of $A_1$ to the middle of $A_2$, likewise for $[B]$). Using these clusters, we can construct a $4$-party entanglement condition $\gamma_{ABC}$ that combines all the two-party entanglement conditions together (Fig.\ref{fig:entanglement_region_picture_4_alt}). Wherever conditions intersect, we relax one of the conditions - this could be from e.g. removing the corridor of low entanglement between $\gamma_A$ and $\gamma_B$ in place of the high entanglement condition in the bulk of $\gamma_C$, or removing the parts where $\gamma_C$ is close to the surface in place of the low entanglement conditions of $\gamma_{A}$ or $\gamma_B$.

\begin{figure}[h]
    \centering
    \begin{tikzpicture}
        \begin{scope}
            \node[anchor=north west,inner sep=0] (image_a) at (0,0)
            {\includegraphics[width=0.99\columnwidth]{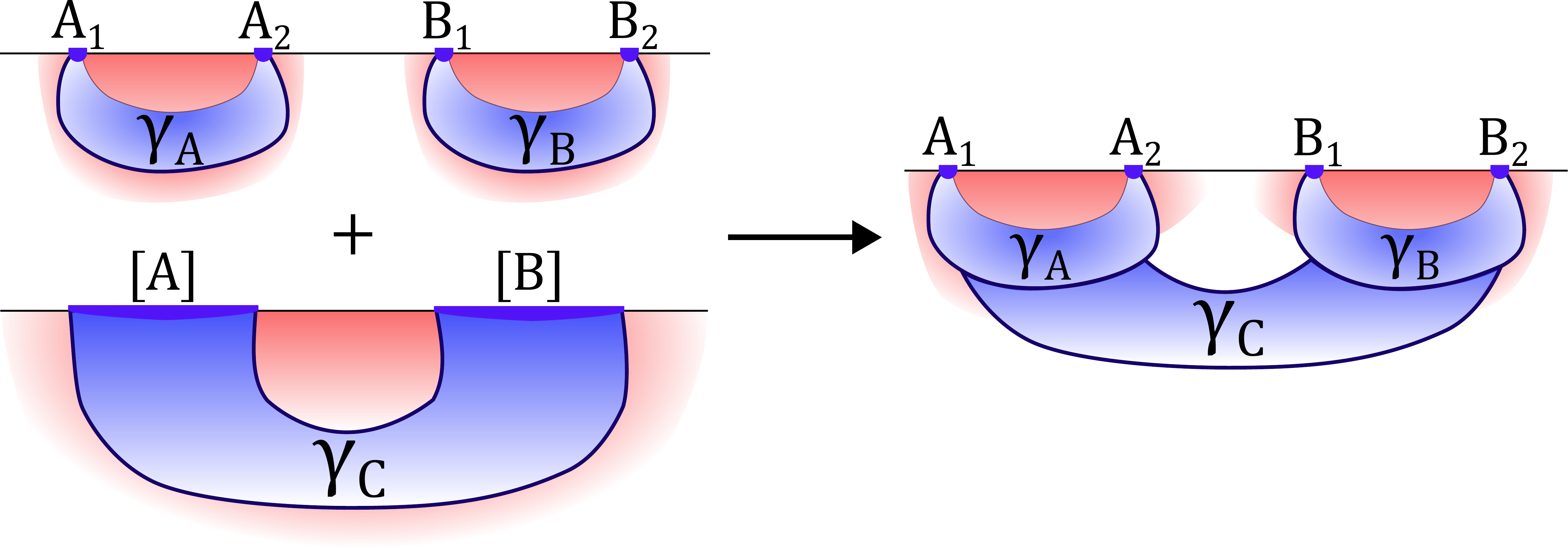}};
        \end{scope}
    \end{tikzpicture}
    \vspace*{-0.2cm}
    \caption{Combining 2-party clusters $\gamma_A, \gamma_B, \gamma_C$ into a 4-party cluster $\gamma_{ABC}$ \textit{(right)}. In the process, the top parts of $\gamma_C$ are replaced by $\gamma_A, \gamma_B$, while the region between $\gamma_A$ and $\gamma_B$ is replaced by $\gamma_C$ - these both correspond to a relaxing of conditions.}
    \label{fig:entanglement_region_picture_4_alt}
\end{figure}

Because of these relaxations, satisfying the combined condition $\gamma_{ABC}$ is easier than satisfying the original three conditions $\gamma_A, \gamma_B, \gamma_C$ independently. Two of those conditions ($\gamma_A, \gamma_B$) correspond to two-party entanglement between subregions with width $w$, which we take to be small, $w\to 0$. The third condition ($\gamma_C$) corresponds to two-party entanglement between subregions $[A],[B]$ of width $1$, so the difficulty of satisfying this condition does not scale with $w$. Therefore, we have obtained a 4-party entanglement condition $\gamma_{ABC}$ whose difficulty scales no worse than two generic 2-party entanglement conditions (plus a constant), implying $\alpha_4 \leq 2\alpha_2$. The argument can be readily generalized to larger $k,\ell$.

\section{Multipartite entanglement in a measurement-only stabilizer circuit} \label{section:theory}

The projective Ising model measurement-only circuit family consists of circuits $C$ acting on sites $S$, consisting of $d$ alternating layers of single-site $X_i$ measurements with probability $1-p$ and two-site $Z_i Z_j$ measurements with probability $p$. The initial state is a product state of $+1$ eigenstates of $X_i$ at each site. Unlike the standard area-law to volume-law transition, measurement-only circuits of this form are known to contain a phase transition between two topologically distinct area-law phases~\cite{Lavasani2021,Ippoliti2021}. Therefore, while the critical point retains logarithmic entanglement entropy scaling, the spatial structure of its entanglement - especially its GME - is unique among measurement-induced phase transitions.

\begin{figure}[h]
    \centering
    \begin{tikzpicture}
        \begin{scope}
            \node[anchor=north west,inner sep=0] (image_a) at (0,0)
            {\includegraphics[width=0.99\columnwidth]{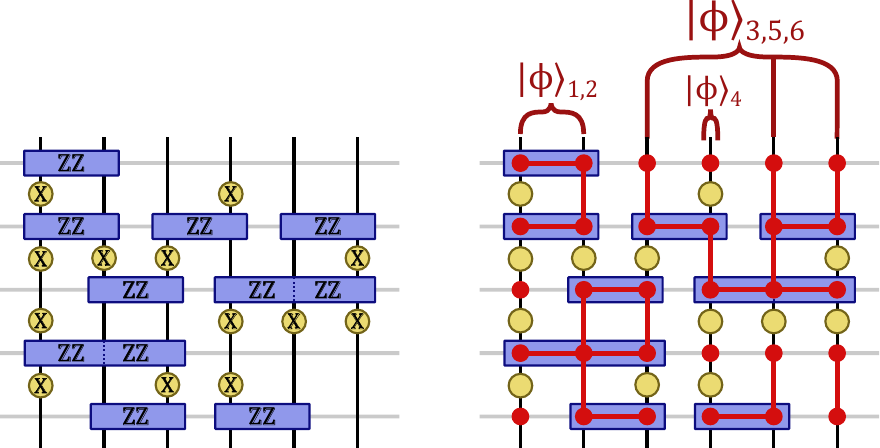}};
            \node [anchor=north west] (note) at (-0.2,-1) {\small{\textbf{a)}}};
            \node [anchor=north west] (note) at (4.2,-1) {\small{\textbf{b)}}};
        \end{scope}
    \end{tikzpicture}
    \caption{Obtaining the loop and percolation models of a realization of a 1+1D measurement-only stabilizer circuit. (a) A particular realization $C$ of the ensemble, consisting of alternating layers of X and ZZ measurements. (b) The percolation model $P_{S,d}(C)$ of the circuit realization, in red, overlayed on top of the original circuit. From this model we can see that sites $\{1,2\}$ form a cat state with 2-party entanglement, and sites $\{3,5,6\}$ form a cat state with 3-party entanglement, while site 4 is unentangled with any neighbor. }
    \label{fig:moc_to_percolation}
\end{figure}

Because each circuit is a member of the Clifford group, its entanglement can be determined from the stabilizers alone~\cite{Bravyi2006}. Additionally, the extra restrictions on the stabilizers leads to a very simple condition for GME. Specifically, we consider the percolation model $P_{S,d}(C)$ on $d$ layers, where each layer is a copy of the complete graph on $S$. Each intralayer bond is open if a $Z_i Z_j$ measurement was applied to the bond's vertices at that particular layer, and each interlayer bond is open if no $X_i$ measurement was applied to the corresponding site at that layer (see Fig.~\ref{fig:moc_to_percolation}b for an example of $P_{S,d}(C)$ on the 1+1D circuit of Fig.~\ref{fig:moc_to_percolation}a). We consider the \textit{cluster surfaces} of that model, i.e. the intersection of any particular connected cluster in $P_{S,d}(C)$ with the final layer. 
The output of $C$ is a product of cat states~\cite{Nahum2020,Lang2020,Sang2020}
\begin{gather}
    \ket{\psi} = \bigotimes_{i=1}^\ell \ket{\phi_i}_{j^i_1 ... j^i_{c(i)}}
\end{gather}
over every cluster surface $\{j^i_1 ... j^i_{c(i)}\}$ (see Appendix~\ref{app:output_of_moc} for a derivation of this result, based on the stabilizers of the circuit).
Each cat state over sites $j_1 ... j_k$ is an equal-amplitude superposition of two site-wise orthogonal product states:
\begin{gather}
    \ket{\phi}_{j_1 ... j_k} = \frac{1}{\sqrt{2}}\left(\ket{\upsilon_1}_{j_1} ... \ket{\upsilon_k}_{j_k} \pm \ket{\overline{\upsilon}_1}_{j_1} ... \ket{\overline{\upsilon}_k}_{j_k}\right)
\end{gather}
such that $\iprod{\upsilon_j}{\overline{\upsilon}_j} = 0$ for all $j$. 
As we can transform any cat state on sites $j_1, ..., j_k$ to a $k$-qubit GHZ state on the same sites using local unitaries~\cite{Bravyi2006}, the amount of $k$-party entanglement shared by those sites under any entanglement monotone $E_k$ is always equal to that of the GHZ state. We can therefore \textit{define} the degree of $k$-party entanglement between sites $j_1,...,j_k$ in a random circuit ensemble as being exactly the probability that those sites form a cluster surface on a given circuit realization. We will refer to this probability as the \textit{hit chance} from now on. Generalizing from sites $j_1,...,j_k$ to subregions $A_1, ..., A_k$, it remains true that any cat state with support in each subregion and nowhere else can be converted to a $k$-qubit GHZ using unitaries local to each subregion. Therefore, we can define the $k$-party entanglement as the average number of cat states with support in each subregion, and nowhere else.

The percolation model of MOCs offers an exact version of the cluster picture from before, as we can directly identify those clusters with the percolation model's connected components. High-entanglement conditions on a cluster $\gamma$ simply become crossing conditions between the ends of $\gamma$, while low-entanglement conditions become prohibitions against crossing - the probability of these conditions being satisfied can then be expressed in terms of Cardy's crossing probabilities~\cite{Cardy1992}.

Technically, the quantity we have defined is not the total GME but the genuine \textit{network} multipartite entanglement~\cite{Navascues2020}. The total GME would also count cases of \textit{indirect} $k$-party entanglement, where the output of the circuit does not contain a $k$-party cat state over the subregions, but each subregion is nonetheless connected to the others by some lower-party state. These connections can be consolidated into a single $k$-party cat state using measurements and classical communication between the parties. However, all such cases of indirect entanglement decay more rapidly over distance than direct entanglement, and therefore do not change the value of $\alpha_k$. 
We cover the case of indirect $k$-party entanglement and its scaling in Appendix~\ref{app:indirect_entanglement}.

We will now focus on the measurement-only stabilizer circuit from Refs.~\onlinecite{Nahum2020} and \onlinecite{Lang2020}, which is on a one-dimensional lattice of $N$ qubits with periodic boundary conditions (Fig.~\ref{fig:moc_to_percolation}a). It is known from the quantum Fisher information~\cite{Lira-Solanilla2025} that these circuits exhibit large entanglement, although this does not necessarily mean that entanglement exists between local subregions - the quadratic QFI at low measurement rates, for example, corresponds to percolation configurations that are too open to host such structures. The percolation model corresponding to this circuit is precisely the bond percolation model on the 2D square lattice, with a critical point at $p=1/2$. All conditions of $k$-party entanglement or mutual information can be approximated to correlation functions of different operators in the corresponding CFT. For example, Ref.~\onlinecite{Sang2020} connected the 2-party entanglement condition to the correlation function of 2 fields with the same scaling dimension as the stress-energy tensor. In Appendix~\ref{app:temperley_lieb}, we generalize this argument to cover entanglement at all parties; the $k$-party entanglement condition will always correspond to the correlation function of $k$ fields, each with the same scaling dimension as the stress-energy tensor, yielding $\alpha_k = 2k$.

\section{Numerical simulation of the 1+1D measurement-only circuit} \label{section:numerics}

\begin{figure}[h]
    \centering
    \begin{tikzpicture}
        \begin{scope}
            \node[anchor=north west,inner sep=0] (image_a) at (0,0)
            {\includegraphics[width=0.49\columnwidth]{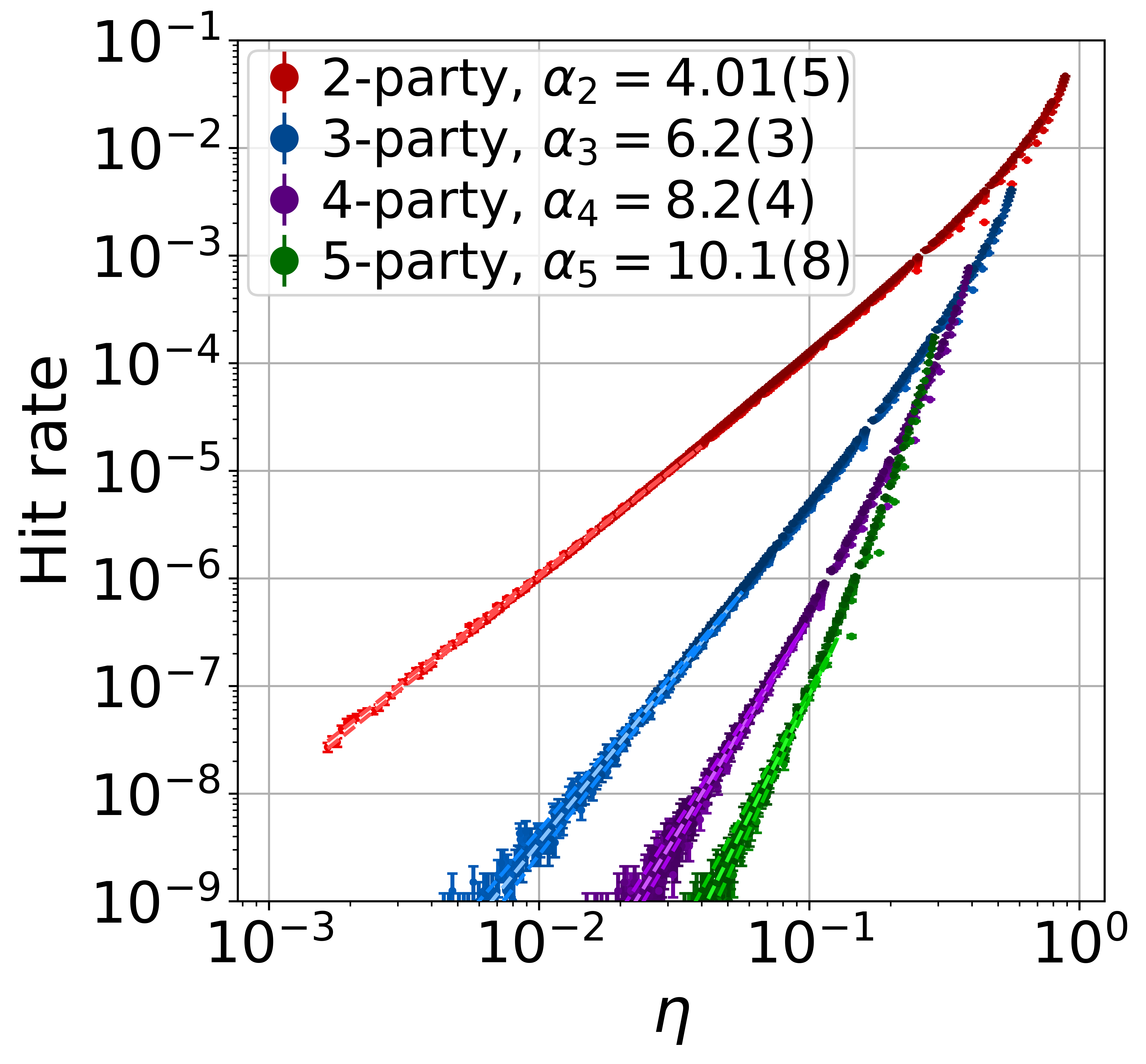}};
            \node [anchor=north west] (note) at (-0.2,0) {\small{\textbf{a)}}};
        \end{scope}
        \begin{scope}[xshift=0.49\columnwidth,yshift=-0.01\columnwidth]
            \node[anchor=north west,inner sep=0] (image_a) at (0,0)
            {\includegraphics[width=0.51\columnwidth]{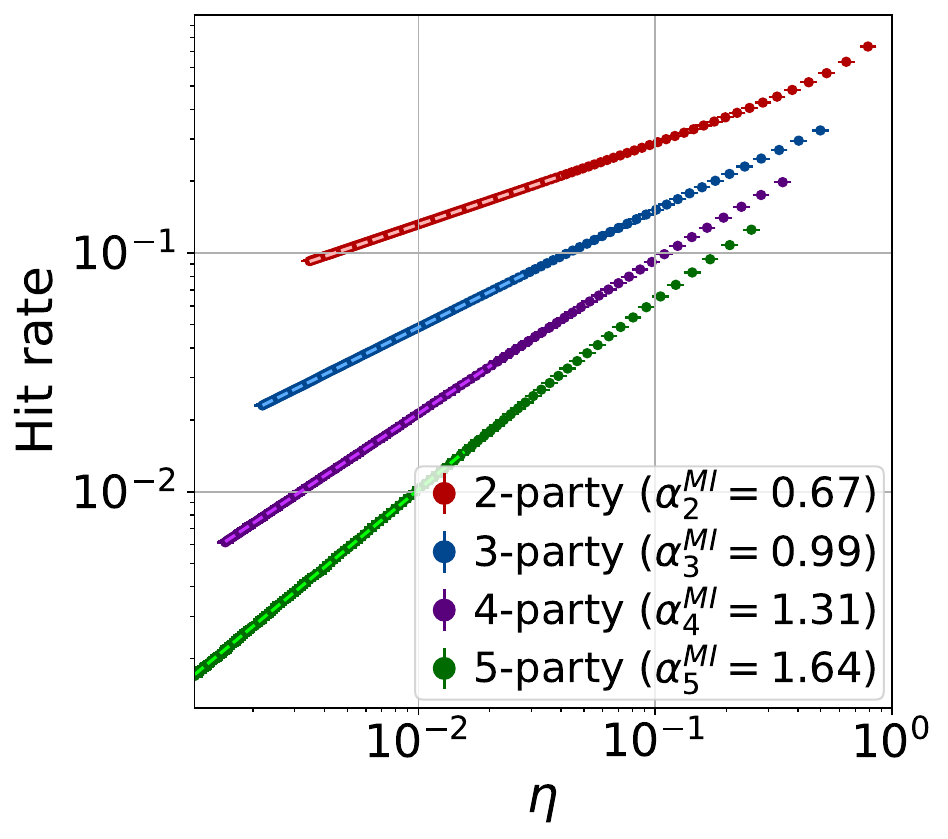}};
            \node [anchor=north west] (note) at (0.2,0) {\small{\textbf{b)}}};
        \end{scope}
    \end{tikzpicture}
    \vspace*{-0.4cm}
    \caption{(a) $k$-party entanglement up to 5 parties over the generalized anharmonic ratio $\eta$. Subregions are 2-16 sites long, with separations up to 48 sites, and darker-colored points correspond to larger width subregions. (b) $k$-party mutual information over $\eta$, for width 8 subregions.}
    \label{fig:gme_over_x}
\end{figure}

In Figure~\ref{fig:gme_over_x}a, we measure the $k$-party entanglement, for $2 \leq k \leq 5$ subregions of equal width and even spacing, for a 512-site circuit of depth 1024, for $4\times 10^9$ circuit iterations. As indirect $k$-party entanglement would only exacerbate short-range effects, and therefore make any fitting of scaling exponents less accurate, we only count cases of direct $k$-party entanglement. For 2 parties, the distance scale to plot the entanglement over is the conformal cross-ratio $\eta = \frac{w_1 w_2}{x_{12}y_{12}}$, where $w_i$ is the width of interval $i$, $x_{ij}$ is the distance from the left endpoint of interval $i$ to the left endpoint of interval $j$, and $y_{ij}$ is the same for right endpoints. All distances $x_{ij},y_{ij}$ and widths $w_i$ are given in terms of chord lengths ${\rm ch}(x)=\frac{N}{\pi}\sin\left(\frac{\pi x}{N}\right)$, due to the periodic boundary conditions.
In this form, $\eta$ acts like the square of the ratio between the geometric mean of the subregion widths and the geometric mean of the subregion distances. We therefore generalize $\eta$ to $k$ subregions by maintaining this rule:
\begin{gather*}
    \eta = \frac{\left(\prod_i w_i\right)^{\frac{2}{k}}}{\left(\prod_{i < j} x_{ij}\right)^{\frac{2}{k(k-1)}}\left(\prod_{i < j} y_{ij}\right)^{\frac{2}{k(k-1)}}}
\end{gather*}
Over the effective distance scale $\lambda$, we have $\eta \sim \lambda^{-2}$, therefore the entanglement scales as $\eta^{\frac{\alpha_k}{2}} = \eta^k$ in the $\eta \rightarrow 0$ limit.

Given a successful instance of $k$-party entanglement between subregions $A_1, ..., A_k$, we can also record what features of the circuit realization led to that entanglement. Generically, we start with a method to describe different circuits in an ensemble $\{\varepsilon\}$ by a graph $G$. The vertices in this graph correspond to degrees of freedom at specific points in time, while edges correspond to circuit elements and have realization-dependent weights $W_\varepsilon$, representing the strength of correlations between neighboring vertices. For the MOC, our graph is the square lattice, and the weight of each edge is the bond openness in the equivalent 1+1D percolation model - a vertical edge will have weight 0 if an $X$ measurement was performed there, 1 if there was no measurement, and vice versa for horizontal edges and $ZZ$ measurements. In order to impose some homogeneity between the two different edges, we perform a convolution by averaging the weight of each horizontal edge with the average weight of the two nearest vertical edges on the next layer, and similarly for the vertical edges. For a more general RQC with measurements, $W_\varepsilon$ could simply be an indicator function that detects measurements on a particular site and layer.

From $W_\varepsilon$ we can define a $k$-party entanglement-weighted graph
\begin{gather}
    \overline{W}[E_k](A_1, ..., A_k)\equiv \frac{\langle  E_k(A_1, ... ,A_k | \varepsilon) \cdot W_\varepsilon \rangle_\varepsilon}{\langle  E_k(A_1, ... ,A_k | \varepsilon)\rangle_\varepsilon}
\end{gather}
The idea is to use the entanglement measure $E_k$ as a statistical weight for $W_\varepsilon$, which in turn captures local entangling/dis-entangling properties of the circuit realization $\varepsilon$. In this way, the $k$-party entanglement-weighted graph describes the local structure of the average circuit realization that has $k$-party entanglement between specific subregions.
In Fig.~\ref{fig:average_cluster}, we show examples of $\overline{W}[E_k]$ for $k=2,3,4$ in the MOC. Remarkably, this graph looks very similar to the clusters described in Section~\ref{section:subadditivity}. We conjecture that $\overline{W}$ can be used to show the connected clusters leading to $k$-party GME (or any other measure of correlations, including $k$-party MI) in other systems as well.

\begin{figure*}[t]
    \centering
    \begin{tikzpicture}
        \begin{scope}
            \node[anchor=north west,inner sep=0] (image_a) at (0,0)
            {\includegraphics[width=0.33\textwidth]{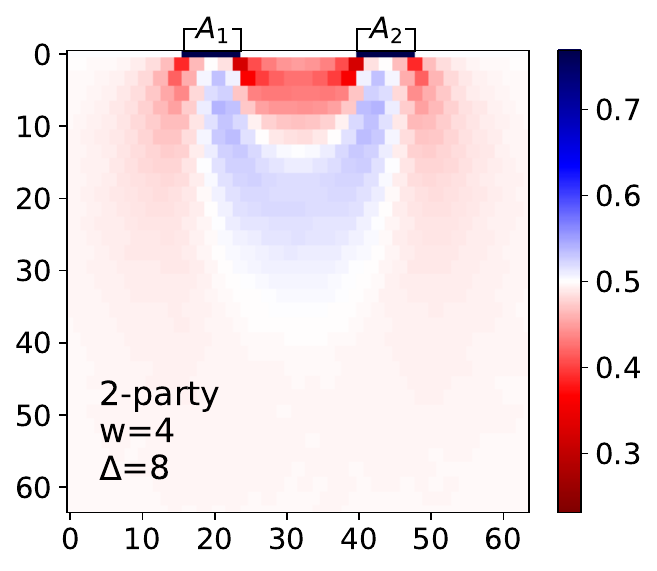}};
            \node [anchor=north west] (note) at (-0.3,0) {\small{\textbf{a)}}};
            \node [anchor=north west, rotate=90] (note) at (-0.4,-3.6) {\small{Circuit depth}};
            \node [anchor=north] (note) at (2.7, -5) {\small{$x$}};
        \end{scope}
        \begin{scope}[xshift=0.33\textwidth]
            \node[anchor=north west,inner sep=0] (image_a) at (0,0)
            {\includegraphics[width=0.33\textwidth]{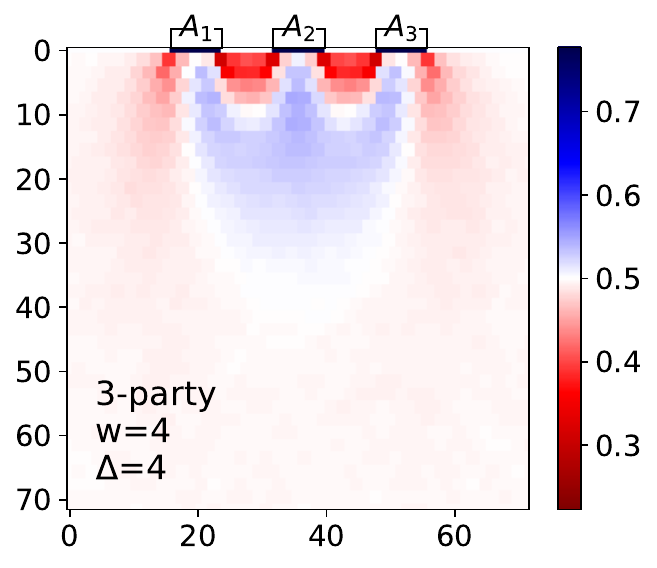}};
            \node [anchor=north west] (note) at (-0.3,0) {\small{\textbf{b)}}};
            \node [anchor=north] (note) at (2.7, -5) {\small{$x$}};
        \end{scope}
        \begin{scope}[xshift=0.66\textwidth]
            \node[anchor=north west,inner sep=0] (image_a) at (0,0)
            {\includegraphics[width=0.33\textwidth]{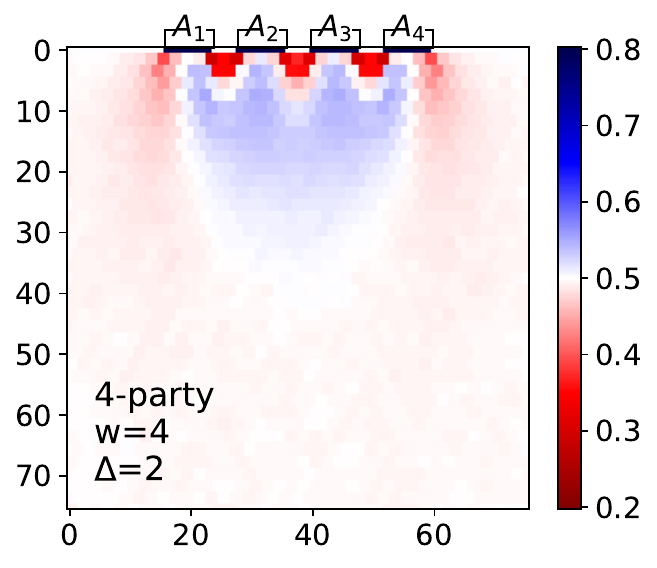}};
            \node [anchor=north west] (note) at (-0.3,0) {\small{\textbf{c)}}};
            \node [anchor=north] (note) at (2.7, -5) {\small{$x$}};
        \end{scope}
    \end{tikzpicture}
    \vspace*{-0.4cm}
    \caption{Entanglement-weighted graph $\overline{W}[E_k]$, given entanglement between $k$ subregions (indicated by the boxes at the top of each figure) of width $w=4$ and different spacings $\Delta$. The weights $W_\varepsilon$ correspond to the average bond openness in the equivalent percolation model (i.e. the density of $ZZ$ measurements and/or the sparsity of $X$ measurements), convolved over a $ZZ$ and $X$ measurement layer. 
    }
    \label{fig:average_cluster}
\end{figure*}

\subsection{$k$-party mutual information} \label{section:k_party_mi}
The $k$-party mutual information can be defined\footnote{There are alternative definitions~\cite{Guo2023,Kumar2025} of the $k$-party mutual information, however, this is the only common definition that does not include any lower-party correlations. For example, if the density matrix is factorizable over a single nontrivial bipartition of the subsets, this is the only measure of $k$-party mutual information that equals zero.} over subregions $A_1, ..., A_k$ as
\begin{gather}\label{eq:k_party_mi}
    I_k(A_1, ..., A_k) = \sum_{j=1}^k (-1)^{j+1} \bbbbbbbb\sum_{B_1,...,B_j \subseteq \{A_1, ..., A_k\}} \bbbbbbbb S(B_1 ... B_j).
\end{gather}
If a cat state contains support on every subregion, as well as the exterior, then $S(B_1 ... B_j) = \ln 2$ for every set of subregions $B_1 ... B_j$. This will always contribute $\ln 2$ to the overall mutual information, and forms its dominant component.
Additionally, there is a correction from cat states living exclusively in the subregions (this is in contrast to the $k$-party GME, where the \textit{only} contributions come from these types of cat state). For these states $S(A_1 ... A_k) = 0$ instead of $\ln 2$. By (\ref{eq:k_party_mi}), the overall contribution to the mutual information from these states is $2 \ln 2$ for even $k$ and zero for odd $k$. 

\begin{figure}[h]
    \centering
    \begin{tikzpicture}
        \begin{scope}
            \node[anchor=north west,inner sep=0] (image_a) at (0,0)
            {\includegraphics[height=0.34\columnwidth]{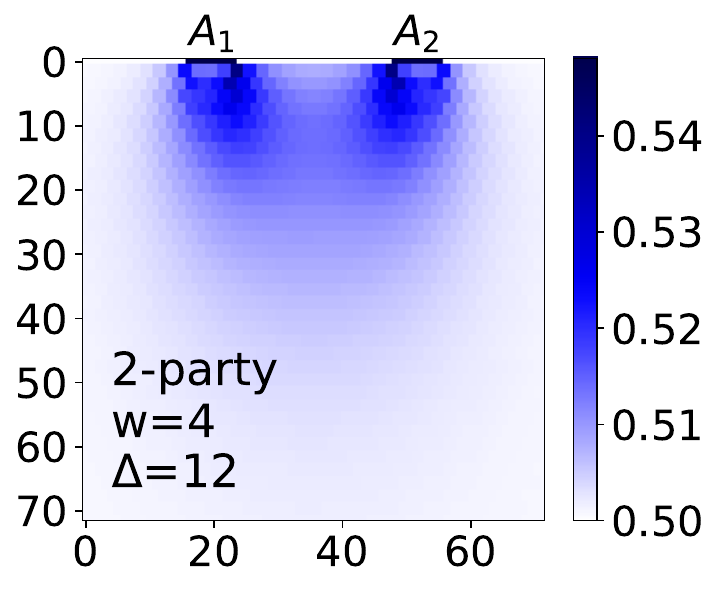}};
            \node [anchor=north west] (note) at (-0.4,0) {\small{\textbf{a)}}};
            \node [anchor=north west, rotate=90] (note) at (-0.4,-2.6) {\small{Circuit depth}};
            \node [anchor=north] (note) at (1.5, -2.75) {\small{$x$}};
        \end{scope}
        \begin{scope}[xshift=0.435\columnwidth]
            \node[anchor=north west,inner sep=0] (image_a) at (0,0)
            {\includegraphics[height=0.34\columnwidth]{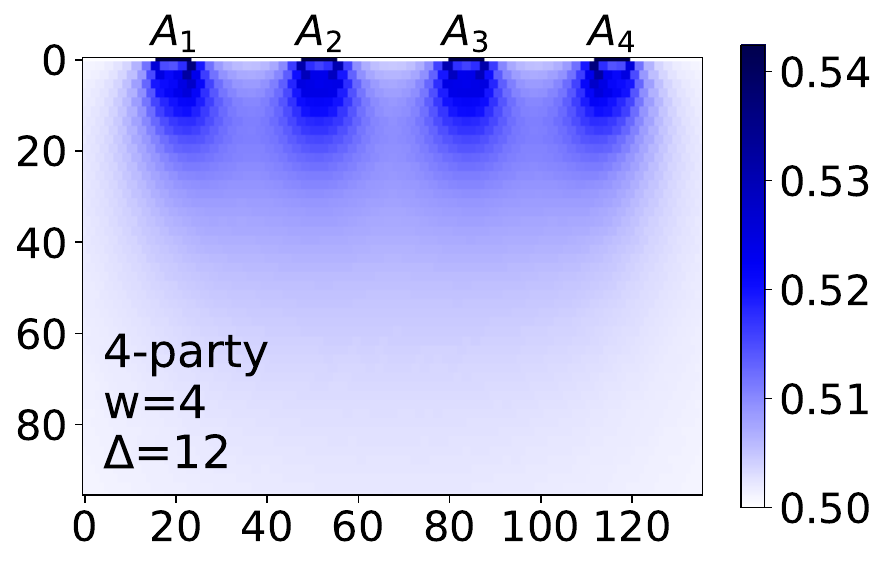}};
            \node [anchor=north west] (note) at (-0.3,0) {\small{\textbf{b)}}};
            \node [anchor=north] (note) at (2, -2.75) {\small{$x$}};
        \end{scope}
    \end{tikzpicture}
    \vspace*{-0.4cm}
    \caption{MI-weighted graph $\overline{W}[I_k]$, given an instance of nonzero mutual information between $k=2,4$ subregions of width $w=4$ and spacing $\Delta=12$. $W_\varepsilon$ is calculated the same way as in Fig.~\ref{fig:average_cluster}. 
    }
    \label{fig:average_cluster_mi}
\end{figure}

The fact that connected clusters in the equivalent percolation model are allowed to extend outside the subregions and still contribute to the mutual information makes its behavior significantly different from that of GME. The MI-weighted graph $\overline{W}[I_k]$ has a more trivial structure, as all low-correlation conditions have been relaxed (Fig.~\ref{fig:average_cluster_mi}). In addition, the correlation functions governing $k$-party mutual information use different fields than those governing $k$-party entanglement.
Based on previously determined correlation functions of the underlying CFT~\cite{Cardy1992,Kleban2006}, we expect the $k$-party mutual information to be governed by a correlation function of $2k$ copies of a boundary operator that is equivalent to the minimal model field $\phi_{1,2}(x)$. This operator corresponds~\cite{Cardy1992} to an exchange in boundary conditions between free and fixed, so there is one operator at the beginning and end of each subregion. The two copies of $\phi_{1,2}(x)$ at the edges of each subregion are then converted, through the operator product expansion, into one copy of $\phi_{1,3}(x)$ for each subregion (as well as a unit operator which only counts clusters that connect on the complement on the subregions~\cite{Kleban2006}). 
The $\phi_{1,3}$ field has a scaling dimension of $h_{1,3} = 1/3$. Therefore, we expect~\cite{Lewellen1992} a boundary field correlation function that scales as $x^{-k/3}$. This prediction is also the proportional generalization of $\alpha_2^{\text{MI}} = d-2+\eta_{||}$, where $d=2$ is the model dimension and the anomalous surface-to-surface scaling exponent $\eta_{||}$ is 2/3 in 2D percolation~\cite{Cardy1984}. The numerical results of Fig.~\ref{fig:gme_over_x}b, run on an $N=4096$, depth 8192 circuit, confirm this prediction, with $\alpha_k^{\rm MI} \approx 0.33k$.

\section{Multiparty entanglement in a 2+1D measurement-only circuit} \label{section:3d}

\begin{figure}[h]
    \centering
    \begin{tikzpicture}
        \begin{scope}[xshift=0\textwidth]
            \node[anchor=north west,inner sep=0] (image_a) at (0.0,0)
            {\includegraphics[width=0.245\textwidth]{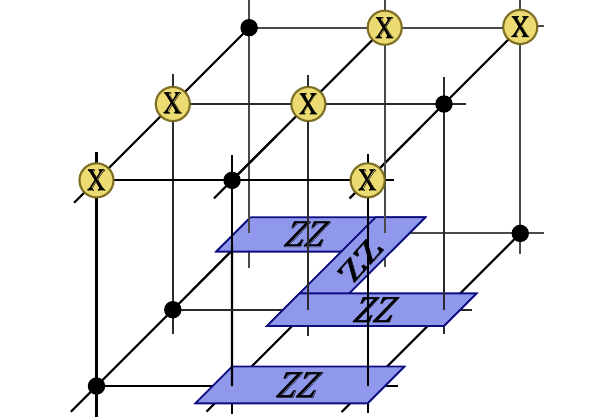}};
            \node [anchor=north west] (note) at (0.6,0) {\small{\textbf{a)}}};
        \end{scope}
        \begin{scope}[xshift=0.245\textwidth, yshift=0.02\textwidth]
            \node[anchor=north west,inner sep=0] (image_a) at (0,0)
            {\includegraphics[width=0.245\textwidth]{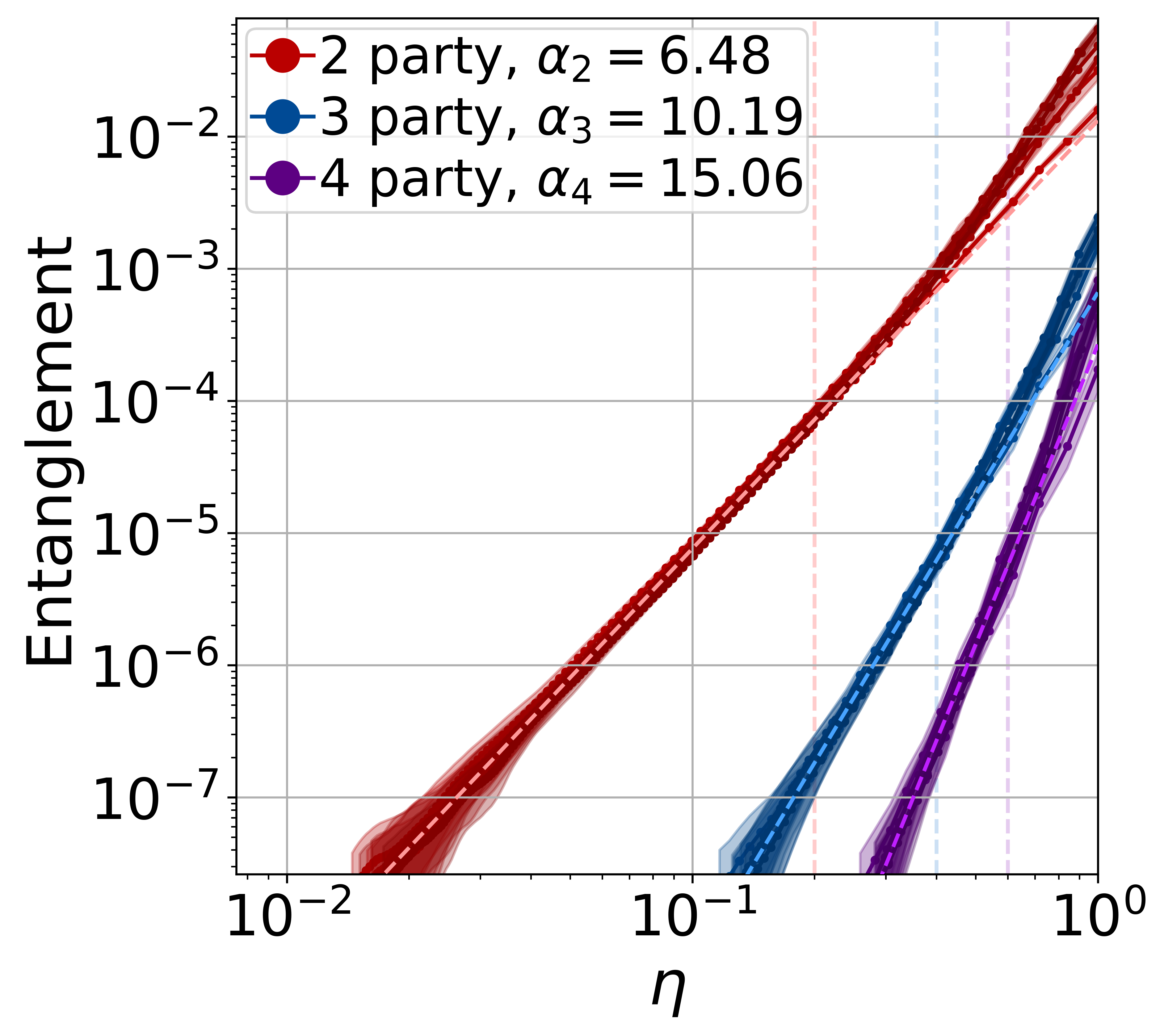}};
            \node [anchor=north west] (note) at (-0.1,0) {\small{\textbf{b)}}};
        \end{scope}
    \end{tikzpicture}
    \begin{tikzpicture}
        \begin{scope}[xshift=0\textwidth]
            \node[anchor=north west,inner sep=0] (image_a) at (0,0)
            {\includegraphics[width=0.245\textwidth]{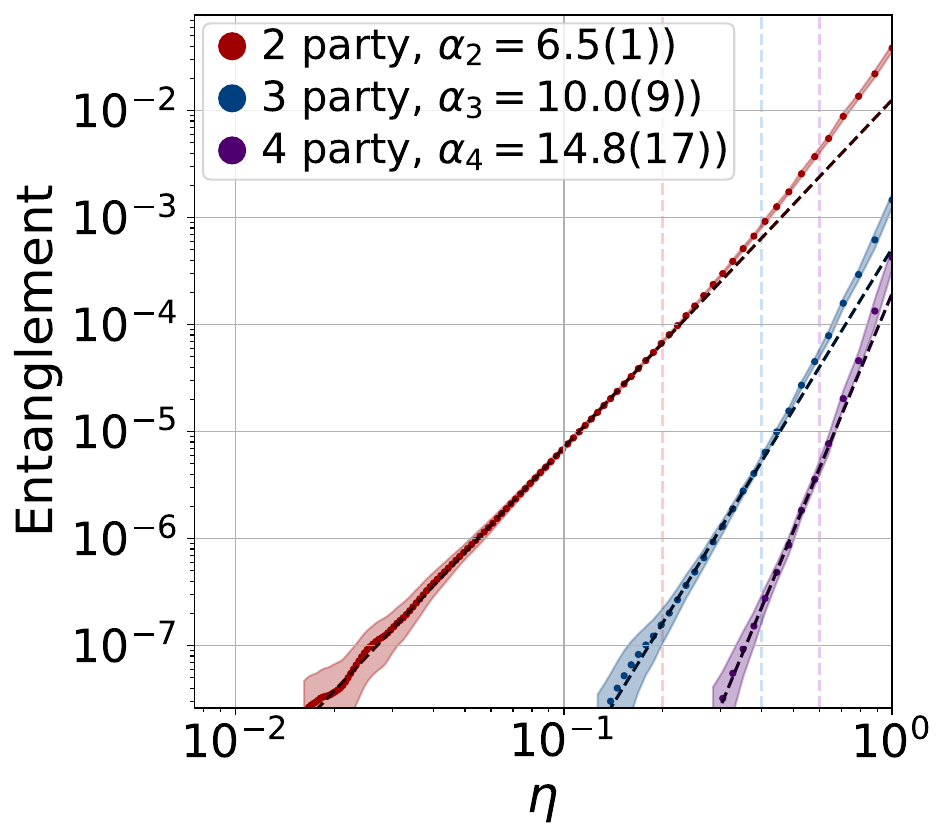}};
            \node [anchor=north west] (note) at (-0.2,0) {\small{\textbf{c)}}};
        \end{scope}
        \begin{scope}[xshift=0.25\textwidth]
            \node[anchor=north west,inner sep=0] (image_a) at (0,0)
            {\includegraphics[width=0.23\textwidth]{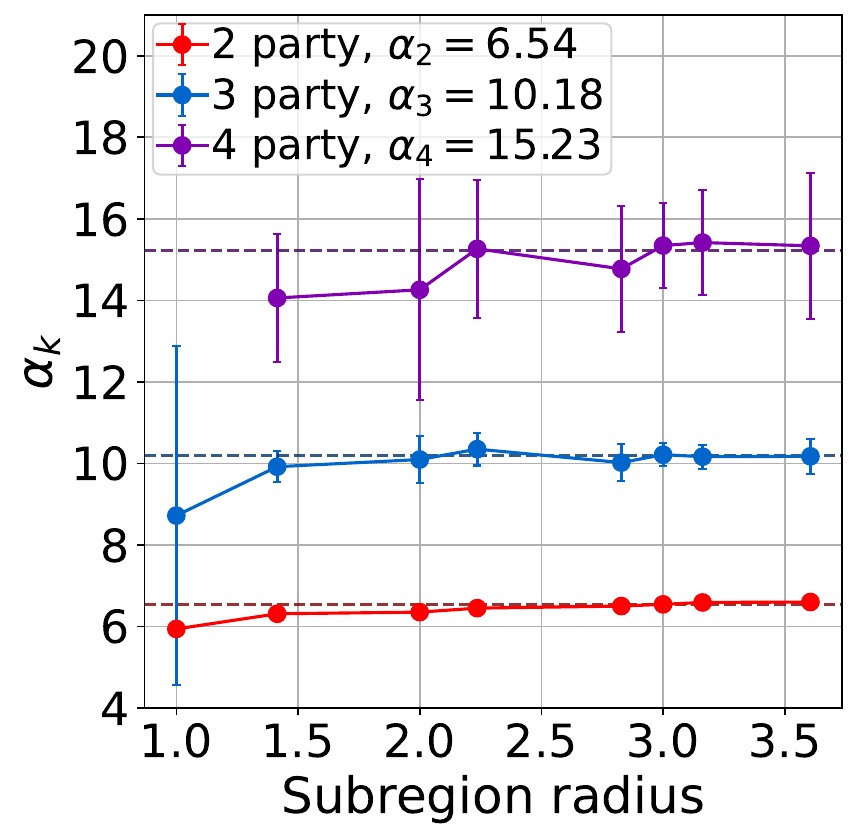}};
            \node [anchor=north west] (note) at (-0.22,0) {\small{\textbf{d)}}};
        \end{scope}
    \end{tikzpicture}
    \vspace*{-0.4cm}
    \caption{(a) Two layers (a $ZZ$ measurement layer and an $X$ measurement layer) of the 2+1D measurement-only circuit. (b) Graph of angle-averaged hit rates for all parties, at all subregion radii above $\sqrt{2}$, for the $L=256$, depth 512 MOC. Darker points indicate larger radii, and shaded areas indicate statistical error from shot noise and performing the angle average. $\alpha_k$ values are taken from fitting data at all subregions. Fitting is only done over $\eta < 0.2, 0.4, 0.6$ for 2, 3 and 4 parties respectively, to avoid short-distance effects. (c) Graph of angle-averaged hit rates of all parties for the subregion inside radius $\sqrt{8}$ (i.e. the 21-site subregion). (d) Estimate of $\alpha_k$ for each subregion of a specific radius. $\alpha_k$ values (dark dotted lines) are taken by averaging the estimates at the five largest radii.}\label{fig:gme_3d}
\end{figure}

In order to make a 2+1D measurement-only circuit that maps to the cubic 3D bond percolation model, we start with an array of qubits in a square lattice with periodic boundary conditions in both the horizontal and vertical directions. On each layer, we apply a ZZ measurement with probability $p$ to each bond on the square lattice, then apply an X measurement with probability $1-p$ to each site (Fig.~\ref{fig:gme_3d}a). This way, the stabilizers of the circuit, and hence the entanglement, are determined by connected clusters in the cubic 3D bond percolation model, which reaches a critical point~\cite{Wang2013} at approximately $p = 0.248812$.

We measure $k$-party entanglement for the 2+1D measurement-only circuit at the critical probability, over an $L\times L$ square lattice to depth $2L$, with $48 \leq L \leq 256$. For each lattice distance $r \leq \sqrt{13}$, we consider the circular subregion centered at the origin $(0,0)$ with radius $r$ consisting of all points fully contained in the interior of the circle. We measure the entanglement between this subregion and $k-1$ others of the same shape, displaced by vectors $(x,y)$ and $(-y,x)$ (see Appendix~\ref{app:2plus1d_kparty_measuring} for details on the exact position of each subregion).

We find evidence that the entanglement has an emergent radial symmetry in the long-distance limit - contours of equal entanglement tend to be circular even if the subregion itself does not have a circular shape (see Fig.~\ref{fig:top_down_entanglement} in Appendix~\ref{app:2plus1d}). The actual definition of the distance scale is complicated by the two-dimensional periodic boundary conditions. We define $\eta$ using the Euclidean distance over chord lengths on the torus:
\begin{gather*}
    \eta = \frac{{\rm ch}(2r)^2}{{\rm ch}(x)^2 + {\rm ch}(y)^2}
\end{gather*}
where ${\rm ch}(z) = \frac{L}{\pi}\sin\left(\frac{\pi z}{L}\right)$ is the chord length over a single dimension. With this metric, graphs of hit rates over $\eta$ at specific subregions tend to vary little over angle, making the angle-averaged hit rate (see Appendix~\ref{app:angle_averaged_entanglement} for the precise definition of this quantity) have relatively little statistical error (Fig.~\ref{fig:gme_3d}b). 

The $k$-party angle-averaged entanglement for each subregion size appears to follow a power law over $\eta$, with the prefactor (i.e. the $\eta \rightarrow 1$ limit of the power law) varying over subregion radius. This is expected, as there are slight differences in the shape of the subregion at each radius, determined by the set of lattice points enclosed by each subregion. A finite size scaling analysis (see Appendix~\ref{app:finite_size_scaling}) suggests that $\alpha_2 \approx 6.2 \pm 0.3, \alpha_3 \approx 9.3 \pm 0.8$ in the limit of infinite system size.

While the 3-party entanglement exponent appears to be slightly superlinear compared to the 2-party exponent, the range of $\eta$ that we fit the 3-party entanglement over was far narrower than the range of the 2-party entanglement fit. Fig.~\ref{fig:gme_3d}b,c) shows that entanglement tends to decay more rapidly over distances at shorter ranges, resulting in a steeper curve at large $\eta$. Therefore, there is a reason to expect that the shorter range of 3-party entanglement leads to a consistently overestimated exponent. The 4-party entanglement is an even more extreme example of this bias.

\begin{figure}[h]
    \centering
    \begin{tikzpicture}
        \begin{scope}
            \node[anchor=north west,inner sep=0] (image_a) at (0,0)
            {\includegraphics[width=0.5\columnwidth]{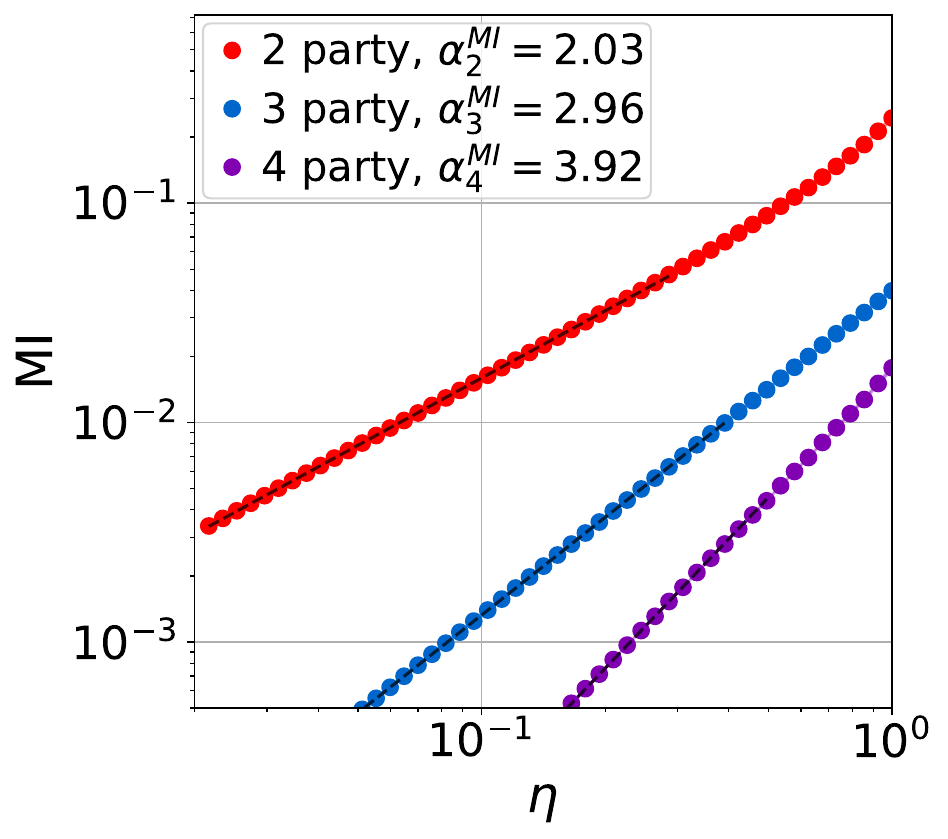}};
            \node [anchor=north west] (note) at (0.1,0) {\small{\textbf{a)}}};
        \end{scope}
        \begin{scope}[xshift=0.51\columnwidth]
            \node[anchor=north west,inner sep=0] (image_a) at (0,0)
            {\includegraphics[width=0.49\columnwidth]{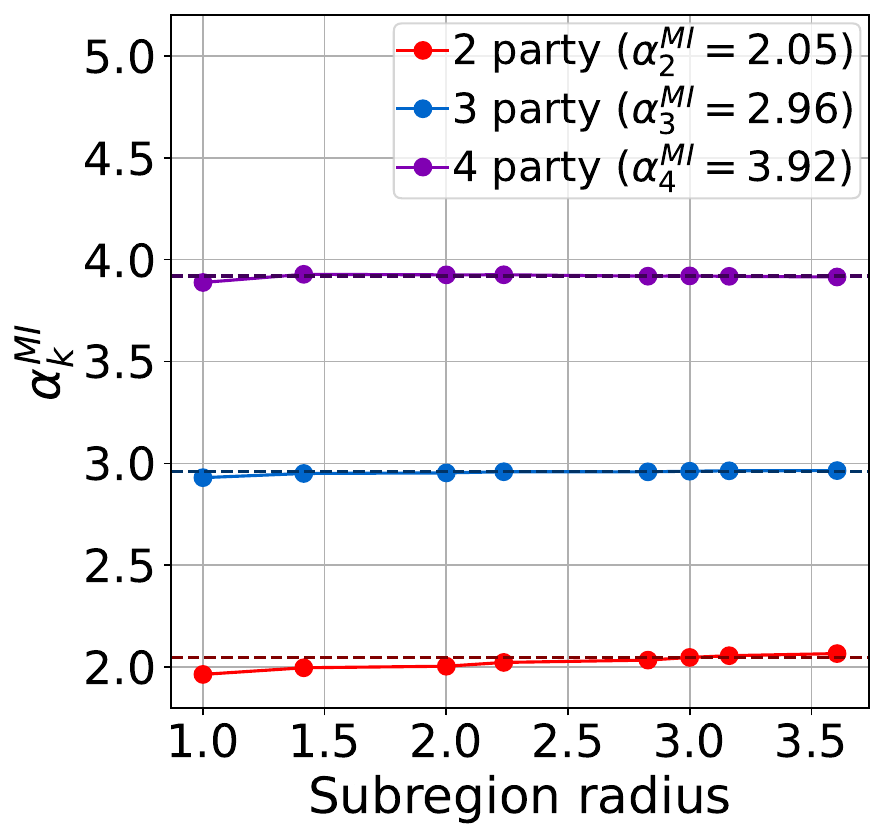}};
            \node [anchor=north west] (note) at (-0.22,0) {\small{\textbf{b)}}};
        \end{scope}
    \end{tikzpicture}
    \vspace*{-0.4cm}
    \caption{(a) 2-4 party angle-averaged mutual information on the $L=256$, depth 512 MOC, for the subregion inside radius $\sqrt{8}$. The estimate of $\alpha_k^{\text{MI}}$, derived from the fit, is in parentheses. (b) Estimates of $\alpha_k^{\text{MI}}$ for each subregion of a specific radius. $\alpha_k^{\text{MI}}$ values (dark dotted lines) are taken by averaging the estimates at the five largest radii. }
    \label{fig:mi_3d}
\end{figure}

We can also numerically measure the $k$-party mutual information between different subregions. Using (\ref{eq:k_party_mi}), this corresponds to the rate of percolation clusters intersecting the surface at all of the specific subregions, even if they intersect the surface outside of those subregions.
For 3D percolation, the anomalous surface-to-surface scaling exponent $\eta_{||}$, which gives a prediction of the 2-party mutual information of $\alpha_2^{\text{MI}} = d-2+\eta_{||}$, has been measured at $0.95(1)$~\cite{Baek2010,Lunt2021}. Our $k$-party mutual information exponents (Fig.~\ref{fig:mi_3d}) are $\alpha_2^{\text{MI}}=2.05(10), \alpha_3^{\text{MI}} = 2.96(4), \alpha_4^{\text{MI}} = 3.92(4)$, using the difference between the lowest and highest measured exponents at each subregion radius as our error bar. These values appear to be consistent with the proportional generalization $\alpha_k^{\text{MI}} = \frac{k}{2}(d-2+\eta_{||}) \approx 0.974k$.

\section{Conclusion}

We have obtained theoretical and numerical evidence that the critical 1+1D measurement-only circuit of Refs.~\onlinecite{Nahum2020,Lang2020,Sang2020} contains genuine $k$-party entanglement with power-law exponent $\alpha_k = 2k$. These exponents easily satisfy the classical dominance (\ref{eq:classical_dominance}) and monotonicity (\ref{eq:monotonicity_of_entanglement}) conjectures from Section~\ref{section:entanglement_exponent_conjectures}, while saturating the subadditivity conjecture (\ref{eq:subadditivity_of_entanglement}) in the 1+1D system. We also have numerical evidence that the natural 2+1D extension contains genuine $k$-party entanglement with power-law exponent $\alpha_k \approx ck$, where $c$ is approximately between 3 and 3.5. As far as we can tell, this specific property has not been measured for the 3D percolation model yet, so the results of this study may be relevant to future works on the 3D percolation model, beyond the specific problem of $k$-party entanglement.

Do the entanglement exponents of other monitored random circuit ensembles - such as the Clifford and Haar circuits - behave in a similar manner? These systems have some similarities to the percolation model - for example, monitored Clifford circuits can be represented as entangled clusters of qubits through a graph state representation~\cite{Anders2006,Lunt2021}, while Haar-random circuits map to a local percolation model either through the zeroth Renyi entropy or under the limit of infinite local Hilbert space dimension~\cite{Skinner2019,Vasseur2019}, although any finite local Hilbert space dimension changes the CFT. However, they are also known to have far greater exponents for $2$-party entanglement~\cite{Sang2020,Avakian2024}. Therefore, their entanglement cannot be equal to a correlation function of $k$ stress tensors like in the measurement-only circuit. In fact, if these systems have \textit{strictly} subadditive exponents, as the numerical evidence of Refs.~\onlinecite{Avakian2024,Lyu2026} suggest, they cannot be represented by a correlation function of $k$ copies of the same field at all. It would therefore be useful to know what correlation functions actually correspond to genuine $k$-party entanglement in these ensembles, and what connections they have to the percolation model.

The shorter-range entanglement of the Clifford and Haar ensembles at low $k$ also raises the question of whether the measurement-only circuit has the longest-range bipartite or tripartite entanglement of any ensemble that is not fine-tuned. 
In Appendix~\ref{app:other_circuits} we provide two random circuit ensembles based on simple generalizations of the MOC that have $\alpha_k < 2k$, but are more fine-tuned than the original MOC. They are not scale invariant and make concessions to their structure or the nature of their gates that could be hard to justify. However, they still leave the possibility that a different version of the circuit ensemble can achieve longer-range entanglement without making such compromises.

Finally, we are interested in expanding and improving the entanglement clusters of section~\ref{section:subadditivity} in future work, such as by studying entanglement-weighted graphs outside the critical point of a system. In addition, the cluster model suggests that genuine $k$-party entanglement between different subregions in a circuit can be described in terms of connections that extend from the top layer into the bulk. This implies there could be a measure for how subregions at different layers can be connected with each other, such that entanglement between subregions on the same layer can be described in terms of combinations of these off-layer measures. In the MOC this measure is clear: two points at different layers are connected if they are connected in the percolation model. A proper general description of this off-layer measure for general systems can help strengthen the entanglement cluster picture. 

\section*{Author Contributions}
This work was developed by J.A. and W.W.-K., with numerics carried out by J.A. under the guidance and direction of W.W.-K. The authors would like to acknowledge some useful discussions with generative AI through Claude and ChatGPT.

\section*{Acknowledgments}
We thank S.~Avakian for bringing to our attention the measurement-only model, and Y.~Saint-Aubin for discussions regarding classical percolation. We also thank Timothy Hsieh and Paul Roux for useful discussions.
W.W.-K. is supported by a grant from the Foundation Courtois, a Chair of the Institut Courtois, a Discovery Grant from NSERC, and a Canada Research Chair.

\bibliographystyle{unsrtnat}
\bibliography{references,references2}

\begin{appendices}

\renewcommand{\thesubsection}{\arabic{subsection}}

\onecolumngrid
\renewcommand{\thesubsection}{\arabic{subsection}}

\section{Output of the measurement-only circuit}\label{app:output_of_moc}
We will start with some definitions of the percolation model, given a circuit $C$ on sites $S$ of depth $d$ consisting of alternating $ZZ$ and $X$ measurements:
\begin{definition}
    For a set of sites $S$ and integer $d$, define the \textit{complete graph cylinder} $G_d(S)$ as the graph whose vertices are $d$ copies of $S$ (we will refer to each copy as a layer in the graph), consisting of \textit{intralayer} edges forming a complete connection of all sites within each layer, and \textit{interlayer} edges connecting a site in a particular layer to its counterparts in adjacent layers.
\end{definition}

\begin{definition}
    The corresponding \textit{percolation model of circuit $C$}, $P_{S,d}(C)$, is the bond percolation model on $G_d(S)$ where each intralayer edge is open if a $Z_i Z_j$ measurement was applied to its vertices at that particular layer, and each interlayer edge is open if no $X$ measurement was applied to the corresponding site at that layer (see Fig.~\ref{fig:moc_to_percolation}). 
\end{definition}

The percolation model $P_{S,d}(C)$ gives rise to clusters in the bulk of $G_d(S)$ connected by open edges. We next consider the intersection of these clusters with the top layer. Specifically,

\begin{definition}
    For a circuit $C$ with percolation model $P_{S,d}(C)$ on graph $G_d(S)$, a \textit{cluster surface} of $C$ is any subset of sites in $S$ corresponding to the intersection of a connected cluster in $P_{S,d}(C)$ with the final layer of the graph.
\end{definition}

As the circuit $C$ belongs to the Clifford group, its stabilizers will consist of some variety of Pauli strings. We can connect those stabilizers to the percolation model $P_{S,d}(C)$ and its cluster surfaces as follows:

\begin{lemma} \label{lemma:percolation_stabilizers}
    The stabilizers of the circuit $C$ are generated by the following operators:
    \begin{itemize}
        \item Two-site operators $Z_i Z_j$, where $i$ and $j$ are any two sites in a cluster surface of $C$,
        \item Multi-site operators $X_{i_1} ... X_{i_k}$, where $i_1...i_k$ is equal to a specific cluster surface of $C$.
    \end{itemize}
    For each such operator, either the operator itself or its negation is the valid generator, depending on the measurement outcomes of the circuit.
\end{lemma}
\begin{proof}
    We can prove this lemma through induction, by analyzing how the stabilizer group evolves (i.e. which sets of operators become valid or invalid generators of the stabilizer group) after applying an $X_i$ or $Z_i Z_j$ measurement. Our initial state has stabilizers $\langle X_i \rangle = +1 \, \forall \, i$, consistent with $N$ unconnected sites in the percolation model. 

    After applying an $X_i$ measurement on site $i$, we have that $\pm X_i$ itself is now a valid stabilizer (depending on the value of the measurement), while $\pm Z_i Z_j$ cannot be a valid stabilizer, for any site $j$. This is consistent with breaking up any clusters connecting site $i$ in the percolation model.
    
    On the other hand, if we apply a $Z_i Z_j$ measurement on sites $i,j$, then $\pm Z_iZ_j$ itself immediately becomes a valid stabilizer (again depending on the value of the measurement). For an operator $\pm X_{a_1}... X_{a_k}$ to remain a valid stabilizer, that operator must commute with $Z_i Z_j$ - therefore, the set $\{a_1 ... a_k\}$ must either contain both $i$ and $j$, or neither. If we had two distinct clusters $\{a_1 ... a_k\}, \{b_1 ... b_\ell\}$ such that $i \in \{a_1 ... a_k\}, j \in \{b_1 ... b_\ell\}$, with $X_{a_1} ... X_{a_k}$ and $X_{b_1} ... X_{b_\ell}$ being valid stabilizers before the measurement, then neither $X_{a_1}...X_{a_k}$ nor $X_{b_1}... X_{b_\ell}$ remain valid stabilizers after the measurement, but their product $X_{a_1} ... X_{b_\ell}$ remains a valid stabilizer. This is consistent with merging two percolation model clusters $\{a_1 ... a_k\},\{b_1 ... b_\ell\}$ through an open bond.

    Because both measurements are consistent with their respective operation in the percolation model, we obtain the two sets of stabilizers from the lemma. If a cluster in the percolation model touches the top layer at $k$ sites, that gives us $k-1$ independent stabilizers of the form $Z_i Z_j$, and $1$ independent stabilizer of the form $X_{i_1} ... X_{i_k}$. Therefore, the number of independent generators defined by the cluster is equal to its size on the top layer, so the number of independent generators defined by the lemma over the whole product state is equal to $|S|$. Hence, these generators form the entire stabilizer group.
\end{proof}
We can immediately identify the states that must be stabilized by such operators as precisely the cat states over each cluster surface, where the precise nature of each cat state (i.e. the $\ket{\upsilon_i}$ states and the superposition phase) depends on the measurement outcomes.

\section{CFT description of $k$-party entanglement}\label{app:temperley_lieb}

\begin{figure}[h]
    \centering
    \vspace*{-0.8cm}
    \begin{tikzpicture}
        \begin{scope}
            \node[anchor=north west,inner sep=0] (image_a) at (0,0)
            {\includegraphics[width=0.79\columnwidth]{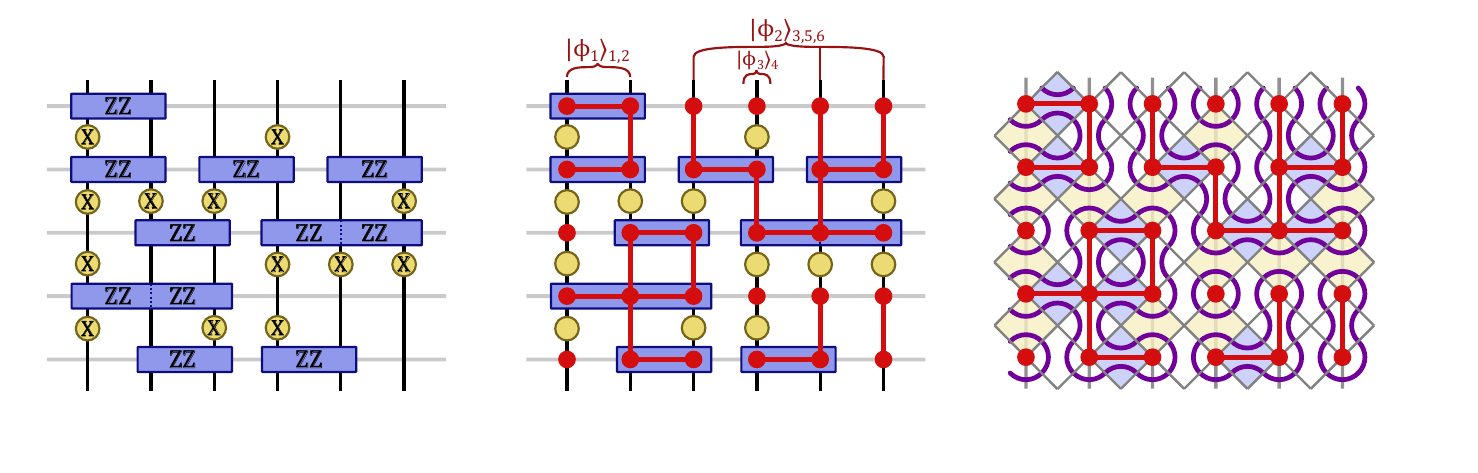}};
            \node [anchor=north west] (note) at (0.1,-0.7) {\small{\textbf{a)}}};
            \node [anchor=north west] (note) at (4.55,-0.7) {\small{\textbf{b)}}};
            \node [anchor=north west] (note) at (9.0,-0.7) {\small{\textbf{c)}}};
        \end{scope}
    \end{tikzpicture}
    \vspace*{-0.9cm}
    \caption{Obtaining the loop model of a realization of a 1+1D measurement-only stabilizer circuit. (a,b), as in Fig.~\ref{fig:moc_to_percolation}, cover the original circuit realization $C$ and its percolation model $P_{S,d}(C)$. In (c), the percolation model $P_{S,d}(C)$ \textit{(red)} is mapped to a loop model generated by the Temperley-Lieb algebra \textit{(purple)}, using the pictorial conversions of $X$ measurements \textit{(yellow)} and $ZZ$ measurements \textit{(blue)} to $\tlx$ symbols, while identities become $\tli$ symbols.}
    \label{fig:moc_to_loop_model}
\end{figure}

The percolation model of the system (Fig.~\ref{fig:moc_to_loop_model}b) can be converted~\cite{Sang2020b,Sang2020} to a loop model (Fig.~\ref{fig:moc_to_loop_model}c) where each cluster is enclosed in loops that are generated by $\tlx$ and $\tli$ symbols. The loop model can be formalized in terms of the Temperley-Lieb algebra~\cite{Temperley1971,Koo1993}. This algebra is generated by elements $e_j$ and weight $\tau$ such that
\begin{align}
    e_j^2 &= \tau e_j \qquad 
    &e_j e_i e_j &= e_j \quad  \text{for } |i-j| = 1 \qquad
    &e_i e_j &= e_j e_i \quad \text{for } |i-j| > 1
\end{align}
We can identify $e_j$ as the $\tlx$ symbol, and $\tau$ is the weight of one loop, which is equal~\cite{Martin1988} to the square root of the number of states $Q$ in the equivalent Potts model.

The Temperley-Lieb algebra gives a method to convert cluster surface conditions into correlation functions of $\tlx$ operators. Specifically, for $k$ parties, the quantity
\begin{gather}
    \dev{}{\tau}\Bigg|_{\tau=1} \sum_{j=1}^{k} (-1)^{k-j}\!\!\!\!\!\!\!\!\!\sum_{a_1 ... a_j \in \{i_1 ... i_k\}} \!\!\!\!\!\!\!\!\langle \tlx_{2a_1} ... \tlx_{2a_j}\rangle =  \dev{}{\tau}\Bigg|_{\tau=1}\left\langle \left(\tlx_{2i_1}\!-1\right) ... \left(\tlx_{2i_k}\!-1\right)\right\rangle 
    \label{eq:general_party_expectation}
\end{gather}
is nonzero iff $i_1, ... ,i_k$ form a cluster surface of $C$.
The proof of this claim is in subsection~\ref{app:entanglement_op_indicator}. Next, we can relate the correlation function of Temperley-Lieb algebra operators to a correlation function of fields on the CFT.
Ref.~\onlinecite{Koo1993} defines a quantity $t_1$ related to the stress energy density in the continuum limit:
\begin{gather}
    t_1 \sim -(\tlx_{2j} + \tlx_{2j-1}) + \text{constant}\\
    T + \overline{T} \propto T_{xx}-T_{yy} \sim t_1 - \langle t_1 \rangle
\end{gather}
We show in section~\ref{app:tl_to_sets} that the $k$-party entanglement between sites $i_1, ..., i_k$ has the same long-distance scaling as a product of $k$ factors of $t_1 - \langle t_1 \rangle$, and therefore has the same long-distance scaling as a product of $k$ factors of the stress-energy tensor. As the stress-energy tensor has scaling dimension 2, if we sort the sites $i_1< ...< i_k$ and there is some characteristic distance scale $\lambda$ such that $i_{j+1}-i_j$ is proportional to $\lambda$, then the $k$-party entanglement between sites $i_1, ..., i_k$ goes as $\lambda^{-2k}$; hence $\alpha_k=2k$.

If we replace each site by an interval, and instead define $\lambda$ as the ratio of interval spacing over interval width, the $\lambda^{-2k}$ scaling should still hold, due to the OPE of boundary stress tensors at $c=0$:
\begin{gather}
    T(a_1) T(a_2) \rightarrow \frac{T(a_1)}{(a_2-a_1)^2} + ...
\end{gather}
Using this, given $k$ intervals $A,B,...$ with endpoints $(a_1, a_2), (b_1, b_2),...$ respectively, we can reduce a $2k$-point correlation function of the form $\langle T(a_1) T(a_2) T(b_1) T(b_2)...\rangle$ to a $k$-point correlation function of the form $\langle T(a_1) T(b_1) ... \rangle$ up to factors of the subregion widths $|a_1 - a_2|, |b_1 - b_2|,$ etc. 

\subsection{Connection between Temperley-Lieb correlation function and entanglement}
\label{app:entanglement_op_indicator}
Here we will review the claim that for $k$ parties, the quantity
\begin{gather*}
    \dev{}{\tau}\Bigg|_{\tau=1} \sum_{j=1}^{k} (-1)^{k-j}\!\!\!\!\!\!\!\!\!\sum_{a_1 ... a_j \in \{i_1 ... i_k\}} \!\!\!\!\!\!\!\!\langle \tlx_{2a_1} ... \tlx_{2a_j}\rangle =  \dev{}{\tau}\Bigg|_{\tau=1}\left\langle \left(\tlx_{2i_1}\!-1\right) ... \left(\tlx_{2i_k}\!-1\right)\right\rangle 
\end{gather*}
is zero for all\footnote{This also includes unphysical loop model configurations that are not planar.} loop model configurations, except the ones satisfying the connection condition $[i_1 ... i_k]$.
Firstly, from Ref.~\onlinecite{Sang2020}, we have that the correlation function derivative
\begin{gather}\label{eq:corr_fn_expression}
    \dev{}{\tau}\bigg|_{\tau=1} \langle \tlx_{J_1} ... \tlx_{J_\ell}\rangle
\end{gather}
is equal to the expected number of loops gained by adding $\tlx_{J_1} ... \tlx_{J_\ell}$ to the end of the model. Each new loop that could potentially be created by the $\tlx_{J_1} ... \tlx_{J_\ell}$ operators corresponds to a connection condition - if a circuit obeys a connection condition $[i_1 ... i_k]$, then that connection condition will contribute a unit to the expression (\ref{eq:corr_fn_expression}) iff $\{i_1 ... i_k\} \subseteq \{J_1 ... J_\ell\}$.

Suppose the state satisfies a connection condition $[j_1 ... j_\ell]$. If the set $\{j_1 ,...,j_\ell\}$ is not completely contained in $\{i_1, ..., i_k\}$, the connection will provide no contributions to the correlation functions at all. Therefore, we will assume that $\{j_1, ..., j_\ell\} \subseteq \{i_1, ..., i_k\}$, and define the (possibly empty) set $\{h_1, ..., h_m\} = \{i_1, ..., i_k\} \setminus \{j_1, ..., j_\ell\}$. 

Then, under the connection condition,
\begin{align}
    \dev{}{\tau}&\Bigg|_{\tau=1}\left\langle \left(\tlx_{2i_1}\!-1\right) ... \left(\tlx_{2i_k}\!-1\right)\right\rangle =  \dev{}{\tau}\Bigg|_{\tau=1} \left \langle \prod_{a=1}^\ell \left(\tlx_{2j_a}\!-1\right) \prod_{b=1}^m \left(\tlx_{2h_b}\!-1\right)\right\rangle\\
    &=\dev{}{\tau}\Bigg|_{\tau=1} \left \langle \sum_{c=0}^\ell (-1)^{\ell-c}\!\!\!\!\!\!\!\!\!\!\sum_{f_1...f_c\subseteq \{j_1 ... j_\ell\}} \!\!\!\!\!\!\! (\tlx_{2f_1} ... \tlx_{2f_c}) \left[(-1)^m + \sum_{d=1}^m (-1)^{m-d} \!\!\!\!\!\!\!\!\!\!\sum_{g_1 ... g_d \subseteq \{h_1 ... h_m\}} \!\!\!\!\!\!\! (\tlx_{2g_1} ... \tlx_{2g_d})\right]\right\rangle
    \\
    &=(-1)^m+\dev{}{\tau}\Bigg|_{\tau=1} \left \langle \sum_{c=0}^\ell (-1)^{\ell-c}\!\!\!\!\!\!\!\!\!\!\sum_{f_1...f_c\subseteq \{j_1 ... j_\ell\}} \!\!\!\!\!\!\! (\tlx_{2f_1} ... \tlx_{2f_c}) \left[\sum_{d=1}^m (-1)^{m-d} \!\!\!\!\!\!\!\!\!\!\sum_{g_1 ... g_d \subseteq \{h_1 ... h_m\}} \!\!\!\!\!\!\! (\tlx_{2g_1} ... \tlx_{2g_d})\right]\right\rangle \\
    &=(-1)^m+ \left[ \sum_{c=0}^\ell (-1)^{\ell-c}\!\!\!\!\!\!\!\!\!\!\sum_{f_1...f_c\subseteq \{j_1 ... j_\ell\}} \!\!\!\!\! (1) \sum_{d=1}^m (-1)^{m-d} \!\!\!\!\!\!\!\!\!\!\sum_{g_1 ... g_d \subseteq \{h_1 ... h_m\}} \!\!\! \left(\delta_{c\ell}+\dev{}{\tau}\Bigg|_{\tau=1}\!\!\!\!\!\!\!\!\left \langle(\tlx_{2g_1} ... \tlx_{2g_d})\right\rangle\right)\right] \\
    &=(-1)^m+ \sum_{d=1}^m (-1)^{m-d} {m \choose d} +  \sum_{c=0}^\ell (-1)^{\ell-c}{ \ell \choose c} \left(\dev{}{\tau}\Bigg|_{\tau=1}\left \langle  \prod_{b=1}^m \left(\tlx_{2h_b}\!-1\right)\right\rangle \right) \\
    &= (-1)^m + \big[(1-1)^m -(-1)^m\big] +  (1-1)^\ell \left(\dev{}{\tau}\Bigg|_{\tau=1}\left \langle  \prod_{b=1}^m \left(\tlx_{2h_b}\!-1\right)\right\rangle \right)
\end{align}
which all cancels out to zero, unless $m=0$. 

\subsection{Relation of Temperley-Lieb algebra operators to stress-energy tensors}\label{app:tl_to_sets}

We will show that the entanglement between sites $i_1, ..., i_k$ has the same long-distance scaling as a product of $k$ factors of $t_1 - \langle t_1 \rangle$, where $t_1$ is a discretized element of the stress energy tensor component $T_{xx}-T_{yy}$. Therefore, $k$-party entanglement has the same long-distance scaling as a product of $k$ factors of the stress-energy tensor. Firstly, we prove the following:
\begin{lemma}\label{lemma:change_identities_to_expectations}
    \begin{gather}
        \dev{}{\tau}\Bigg|_{\tau = 1} \big \langle (\tlx_{I_1} - \langle \tlx_{I_1}\rangle) ... (\tlx_{I_k} - \langle \tlx_{I_k}\rangle)\big\rangle = \dev{}{\tau}\Bigg|_{\tau = 1} \big \langle (\tlx_{I_1} - 1) ... (\tlx_{I_k} - 1)\big\rangle
    \end{gather}
\end{lemma}
\begin{proof}
    Expanding the product on the LHS,
    \begin{gather}
        \dev{}{\tau}\Bigg|_{\tau = 1} \big \langle (\tlx_{I_1} - \langle \tlx_{I_1}\rangle) ... (\tlx_{I_k} - \langle \tlx_{I_k}\rangle)\big\rangle = \dev{}{\tau}\Bigg|_{\tau = 1} \sum_{j=0}^k (-1)^{k-j}\!\!\!\!\!\!\! \sum_{A_1 ... A_j \in \{I_1 ... I_k\}} \!\!\!\!\!\!\langle \tlx_{A_1} ... \tlx_{A_j}\rangle \!\!\prod_{B \notin \{A_1 ...A_j\}} \!\!\!\!\langle \tlx_B\rangle
    \end{gather}
    Because $\langle \tlx_{J_1} ... \tlx_{J_\ell}\rangle_{\tau = 1} = 1$ for any set of $\tlx$ operators, applying the product rule gives, for $k \geq 2$:
    \begin{align}
        \text{LHS} &= \dev{}{\tau}\Bigg|_{\tau=1} \sum_{j=0}^k (-1)^{k-j}\!\!\!\!\!\!\! \sum_{A_1 ... A_j \in \{I_1 ... I_k\}} \!\!\!\left(\langle \tlx_{A_1} ... \tlx_{A_j}\rangle + \!\!\sum_{B \notin \{A_1 ... A_j\}} \!\!\!\!\langle \tlx_B\rangle\right)\n 
        &= \dev{}{\tau}\Bigg|_{\tau=1} \sum_{j=0}^k (-1)^{k-j}\!\!\!\!\!\!\! \sum_{A_1 ... A_j \in \{I_1 ... I_k\}} \!\!\!\langle \tlx_{A_1} ... \tlx_{A_j}\rangle + \sum_{j=0}^k (-1)^{k-j} {k \choose j} \frac{k-j}{k}\!\!\!\!\sum_{B \in \{I_1 ... I_k\}} \!\!\!\!\langle \tlx_B \rangle \n
        &= \dev{}{\tau}\Bigg|_{\tau = 1} \big \langle (\tlx_{I_1} - 1) ... (\tlx_{I_k} - 1)\big\rangle + \sum_{j=0}^{k-1} (-1)^{k-j} {k-1 \choose j} \!\!\sum_{B \in \{I_1 ... I_k\}} \!\!\!\!\langle \tlx_B\rangle \n
        &= \dev{}{\tau}\Bigg|_{\tau = 1} \big \langle (\tlx_{I_1} - 1) ... (\tlx_{I_k} - 1)\big\rangle + (-1)^{k} (1-1)^{k-1} \!\!\sum_{B \in \{I_1 ... I_k\}} \!\!\!\!\langle \tlx_B\rangle \n
        &= \dev{}{\tau}\Bigg|_{\tau = 1} \big \langle (\tlx_{I_1} - 1) ... (\tlx_{I_k} - 1)\big\rangle
    \end{align}
\end{proof}
The $k$-party entanglement $E(i_1, ..., i_k)$ is equal to the LHS of Lemma~\ref{lemma:change_identities_to_expectations}. Next, we separate even operators $\tlx_{2j}$ from odd operators $\tlx_{2j+1}$:
\begin{lemma} \label{lemma:continuum_limit}
    Assuming we are in the limit of large distances (i.e. $|i_a-i_b| \gg 1$ for all $1 \leq a,b \leq k$),
    \begin{align}
        \dev{}{\tau}\Bigg|_{\tau = 1} &\left\langle \prod_{j=1}^{k}\big(\!\tlx_{2i_j} + \tlx_{2i_j+1} - \langle \tlx_{2i_j}\rangle - \langle \tlx_{2i_j+1}\rangle\big)\right\rangle =\n
        \dev{}{\tau}\Bigg|_{\tau = 1}& \left\langle \prod_{j=1}^k \big(\!\tlx_{2i_j} - \langle \tlx_{2i_j}\rangle \big)\right\rangle + 
        \left\langle \prod_{j=1}^k \big(\!\tlx_{2i_j+1} - \langle \tlx_{2i_j+1}\rangle \big)\right\rangle \label{eq:odd_even_separation}
    \end{align}
\end{lemma}
\begin{proof}
    We take the limit of large distances to ensure all $\big(\!\tlx_{2i_j} + \tlx_{2i_j+1} - \langle \tlx_{2i_j}\rangle - \langle \tlx_{2i_j+1}\rangle\big)$ terms commute, and therefore can be applied to the top layer simultaneously. Therefore, for the rest of this proof we assume that any two operators of the form $\tlx_{2i_a (+1)}$, $\tlx_{2i_b(+1)}$, with $a \neq b$, commute.
    
    We start with an observation on the nature of the Temperley-Lieb algebra: if a loop is closed by applying some set of $\tlx_j$ operators at the top layer, then all the $j$'s must be of the same parity. In other words, if we apply a group of $\tlx_j$ operators to the top layer, then the loops that get created can be separated into those that were closed with only even-parity operators and those that were closed with only odd-parity operators:
    \begin{gather}
        L(\tlx_{2j_1} ... \tlx_{2j_k} \tlx_{2\ell_1+1} ... \tlx_{2\ell_m+1}, w) = L(\tlx_{2j_1} ... \tlx_{2j_k},w) + L( \tlx_{2\ell_1+1} ... \tlx_{2\ell_m+1}, w)
    \end{gather}
    Therefore,
    \begin{align}
        \dev{}{\tau} \Bigg|_{\tau=1} \langle \tlx_{2j_1} ... \tlx_{2j_k} \tlx_{2\ell_1+1} ... \tlx_{2\ell_m+1}\rangle &= \dev{}{\tau} \Bigg|_{\tau=1} \langle \tlx_{2j_1} ... \tlx_{2j_k}\rangle + \langle \tlx_{2\ell_1+1} ... \tlx_{2\ell_m+1}\rangle\n 
        &= \dev{}{\tau} \Bigg|_{\tau=1} \langle \tlx_{2j_1} ... \tlx_{2j_k}\rangle \langle \tlx_{2\ell_1+1} ... \tlx_{2\ell_m+1}\rangle
    \end{align}
    Likewise, 
    \begin{align}
        \dev{}{\tau} \Bigg|_{\tau=1} &\bigg\langle (\tlx_{2j_1}-\langle\tlx_{2j_1}\rangle) ... (\tlx_{2j_k}-\langle\tlx_{2j_k}\rangle) (\tlx_{2\ell_1+1}-\langle \tlx_{2\ell_1+1}\rangle) ... (\tlx_{2\ell_m+1}-\langle \tlx_{2\ell_m+1}\rangle)\bigg\rangle \n 
        &= \dev{}{\tau} \Bigg|_{\tau=1} \bigg\langle (\tlx_{2j_1}-\langle\tlx_{2j_1}\rangle) ... (\tlx_{2j_k}-\langle\tlx_{2j_k}\rangle) \bigg\rangle \bigg\langle (\tlx_{2\ell_1+1}-\langle \tlx_{2\ell_1+1}\rangle) ... (\tlx_{2\ell_m+1}-\langle \tlx_{2\ell_m+1}\rangle)\bigg\rangle
    \end{align}
    In fact, this term should be zero as long as both $k > 0$ and $m > 0$, as 
    \begin{gather}
        \langle \tlx_{a_1} ... \tlx_{a_c}\rangle \Bigg|_{\tau = 1} = \bigg\langle \tau^{L(a_1 ... a_c, w)}\bigg\rangle\Bigg|_{\tau=1} = 1\n
        \bigg\langle (\tlx_{a_1}-\langle\tlx_{a_1}\rangle) ... (\tlx_{a_c}-\langle\tlx_{a_c}\rangle) \bigg\rangle\Bigg|_{\tau=1} = (1-1)^c = 0
    \end{gather}
    for all $c > 0$. Returning to the LHS of (\ref{eq:odd_even_separation}), 
    \begin{align}
        \dev{}{\tau}\Bigg|_{\tau = 1} \left\langle \prod_{j=1}^{k}\big(\!\tlx_{2i_j} + \tlx_{2i_j+1} - \langle \tlx_{2i_j}\rangle - \langle \tlx_{2i_j+1}\rangle\big)\right\rangle = \dev{}{\tau}\Bigg|_{\tau = 1} \sum_{\vec{b} \in \{0,1\}^k}\left\langle \prod_{j=1}^{k}\big(\!\tlx_{2i_j+b_j} - \langle \tlx_{2i_j+b_j}\rangle \big)\right\rangle,
    \end{align}
    we therefore have that every term in the sum over $\vec{b}$ must be zero, unless either all the $b_j$'s are zero, or all the $b_j$'s are one, yielding the RHS of (\ref{eq:odd_even_separation}).
\end{proof}
In the large distance limit, the quantities $\dev{}{\tau} \big|_{\tau=1} \left\langle \prod_{j=1}^k \big(\!\tlx_{2i_j} - \langle \tlx_{2i_j}\rangle \big)\right\rangle $ and $\dev{}{\tau} \big|_{\tau=1}\left\langle \prod_{j=1}^k \big(\!\tlx_{2i_j+1} - \langle \tlx_{2i_j+1}\rangle \big)\right\rangle$ have the same scaling laws, and both quantities are nonnegative, so either term must have the same scaling law as the sum. Since the former term is $E(i_1, ..., i_k)$ and the sum is the expectation value of $k$ copies of $t_1 - \langle t_1\rangle$, that means $E(i_1, ..., i_k)$ has the same scaling law as the expectation value of $k$ factors of the stress-energy tensor. Therefore, for sites separated by distances proportional to $x$, $E(i_1, ..., i_k)$ should go as $x^{-2k}$.

\section{Extra details on simulation algorithms}\label{app:algorithms}

\subsection{Simulating the 2+1D percolation model} \label{app:2plus1d}

As the loop model does not exist in 2+1D percolation, we cannot use it to simplify our algorithm, and must rely on a more general percolation algorithm. 

At each site, we store the index of the cluster in the percolation model it is currently a part of. We note that there is at most $N$ clusters on a particular layer, so with some index recycling we can bound our cluster indices by an $O(N)$ quantity, which we will call $C(N)$ - for our algorithm we set $C(N) = 2N$. When a measurement merges clusters together, the sites in both clusters should adopt the lower cluster index. Each layer consists of a layer of nearest-neighbor $ZZ$ measurements, where clusters may be merged together but not created, and a layer of single-site $X$ measurements, where new clusters may be created. 

Throughout the $ZZ$ measurements we store two auxiliary objects. The first is a $C(N)$-sized ``in-use" array which specifies which cluster indices were either not being used when the $ZZ$ layer was applied, or which will be merged out by the $ZZ$ measurements. This object is assembled in the $ZZ$ layer and then used to determine recyclable cluster indices in the $X$ measurement layer, with both operations contributing an $O(N)$ cost towards the layer. 

The second is a set of cluster-merging trees, where each tree contains the indices of all clusters that must be merged into the index at the root of the tree, due to $ZZ$ measurements connecting each cluster in the tree with its parent. This is again represented by a $C(N)$-sized array whose entries are the cluster index of the immediate parent of the node - or, if the cluster should not be merged out, the array's entry is the cluster index itself. When a $ZZ$ measurement is applied to two sites of different clusters, the node of the larger cluster index gets assigned to the node of the smaller cluster index. If the node of the larger cluster index already has a different parent, however, we must merge the trees of both parents together, by tracing both trees to their root and comparing the cluster indices of the roots. After all $ZZ$ measurements have been taken, we officially merge all clusters by assigning each cluster index to the one at the root of its respective tree.

There are two opportunities for this process to acquire a superlinear cost in $N$ - firstly, when two trees need to be merged (a process that might be called $\Theta(N)$ times over a single layer), the action of tracing both trees to the root takes a cost equal to the sum of the tree depths. Therefore, an average tree depth of $\omega(1)$ implies a superlinear contribution to the computational cost per layer. Next, the process of assigning each cluster to its root may be superlinear in $N$, although it only has to be called once per layer. While it is difficult to determine the true average cost of either process, we can still measure the costs numerically. In Fig.~\ref{fig:costs_of_3d}, we graph the approximate per-site operational cost of both parts of the algorithm over $N$, as well as the per-layer, per-site time cost of the whole algorithm. This evidence implies that for large $N$, the per-site operational cost of the processes approaches a constant, or a quantity that is almost constant. 

\begin{figure}[h]
    \centering
    \begin{tikzpicture}
        \begin{scope}
            \node[anchor=north west,inner sep=0] (image_a) at (0,0)
            {\includegraphics[width=0.5\columnwidth]{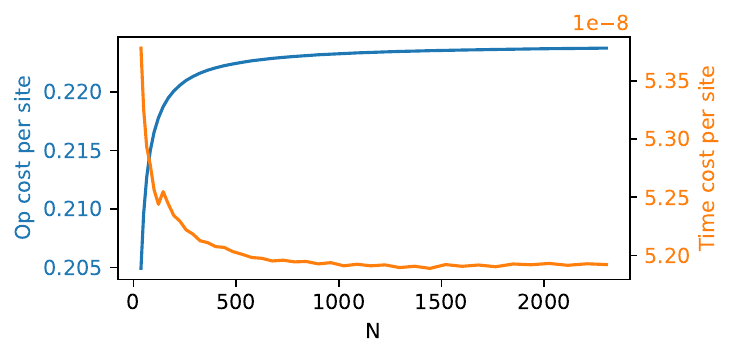}};
        \end{scope}
    \end{tikzpicture}
    \vspace*{-0.4cm}
    \caption{Time and operational cost analysis of the percolation cluster algorithm used for the 2+1D model. We plot the average time cost per site of updating a single layer \textit{(orange)} and the approximate operational cost of all potentially superlinear processes \textit{(blue)} over system size $N$, for $L\times L$ systems, with $6 \leq L \leq 48$.}
    \label{fig:costs_of_3d}
\end{figure}

\begin{figure}[h]
    \centering
    \begin{tikzpicture}
        \begin{scope}
            \node[anchor=north west,inner sep=0] (image_a) at (0,0)
            {\includegraphics[width=0.8\columnwidth]{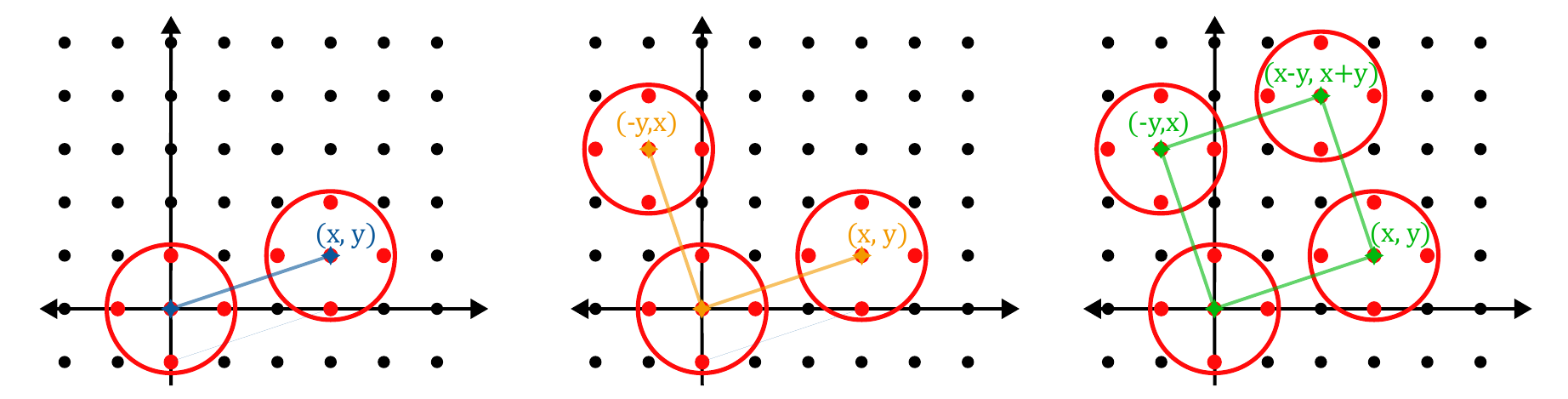}};
        \end{scope}
    \end{tikzpicture}
    \vspace*{-0.4cm}
    \caption{Subregions to detect entanglement on in the 2+1D measurement-only stabilizer circuit. We check if the circular subregion centered at the origin of a given radius has 2-party entanglement \textit{(blue)} with the subregion centered at $(x,y)$, 3-party entanglement \textit{(orange)} with the subregions at $(x,y)$ and $(-y,x)$, or 4-party entanglement \textit{(green)} with the subregions at $(x,y)$, $(-y,x)$ and $(x-y,x+y)$.}\label{fig:moc_3d}
\end{figure}

\subsection{Measuring $k$-party entanglement on the 2+1D MOC} \label{app:2plus1d_kparty_measuring}

Like in the 1+1D case, the necessary and sufficient condition for $k$-party entanglement to exist between $k$ disjoint subregions is that a cluster surface exists in the equivalent percolation model, i.e. a cluster intersects the surface inside each of the subregions, and does not intersect the surface anywhere else.

While in the 1+1D case we had to arrange each interval on a line, in 2+1D space we can position our subregions anywhere on the 2D surface of the system. As long as all distances in the system are governed by a single scale $x$, the entanglement should have the same scaling law over $x$. In our case, we choose to measure up to $4$-party entanglement by arranging our subregions on corners of a square.

Specifically, for each distance vector $(x,y)$, we measure the 2-party entanglement between the subregions centered at $(0,0)$ and $(x,y)$, the 3-party entanglement between $(0,0)$, $(x,y)$ and $(-y,x)$, as well as the 4-party entanglement between $(0,0)$, $(x,y)$, $(-y, x)$ and $(x-y, x+y)$, effectively filling out a square of side length $\sqrt{x^2+y^2}$ with each additional subregion (Fig.~\ref{fig:moc_3d}). To make sure the subregions do not intersect, we only look at the case $x^2+y^2 > (2r)^2$, roughly corresponding to $\eta < 1$.

\begin{figure}[h]
    \centering
    \begin{tikzpicture}
        \begin{scope}
            \node[anchor=north west,inner sep=0] (image_a) at (0,0)
            {\includegraphics[width=0.49\columnwidth]{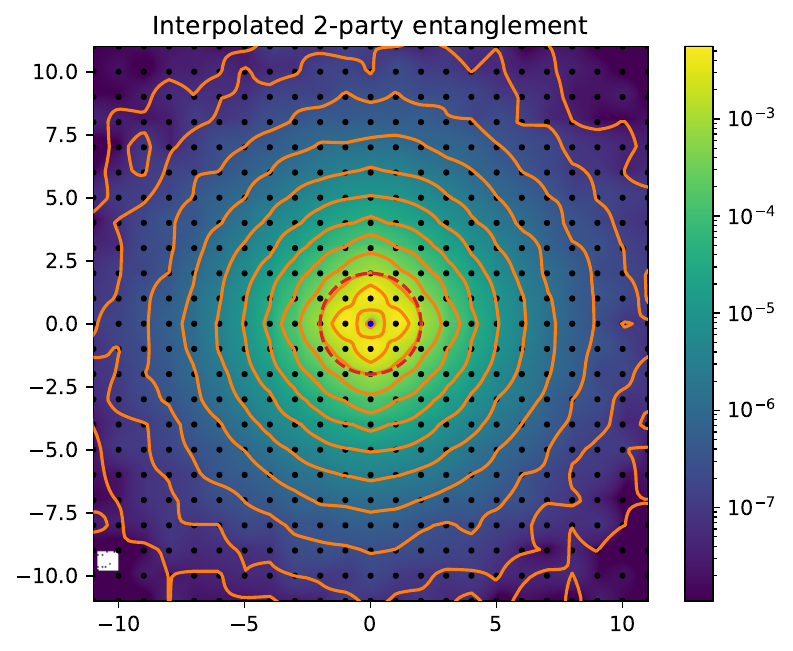}};
            \node [anchor=north west] (note) at (0,0) {\small{\textbf{a)}}};
        \end{scope}
        \begin{scope}[xshift=0.5\columnwidth]
            \node[anchor=north west,inner sep=0] (image_a) at (0,0)
            {\includegraphics[width=0.5\columnwidth]{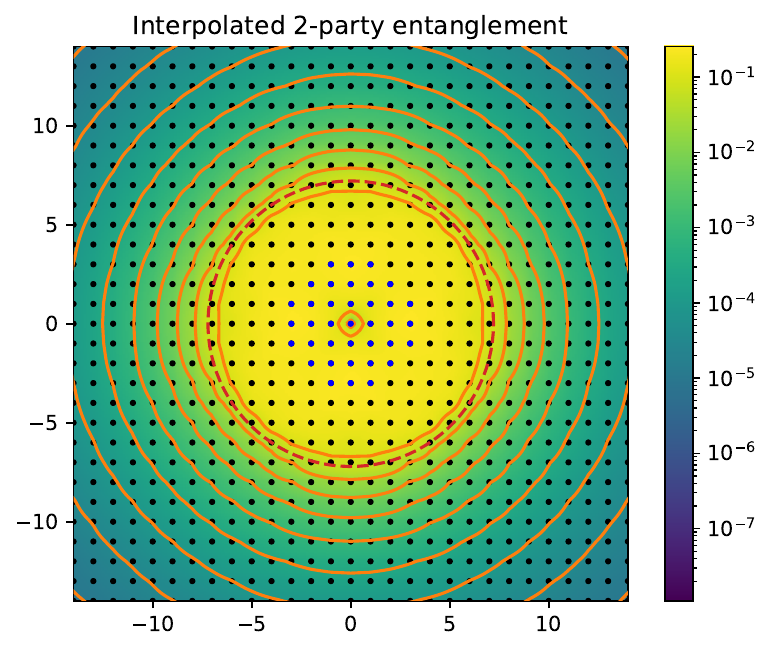}};
            \node [anchor=north west] (note) at (0,0) {\small{\textbf{b)}}};
        \end{scope}
    \end{tikzpicture}
    \begin{tikzpicture}
        \begin{scope}
            \node[anchor=north west,inner sep=0] (image_a) at (0,0)
            {\includegraphics[width=0.5\columnwidth]{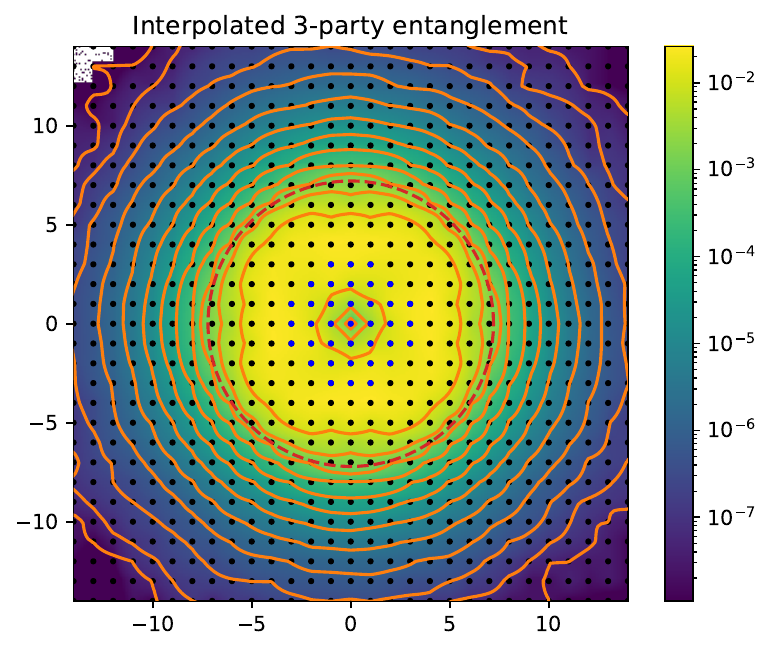}};
            \node [anchor=north west] (note) at (0,0) {\small{\textbf{c)}}};
        \end{scope}
        \begin{scope}[xshift=0.5\columnwidth]
            \node[anchor=north west,inner sep=0] (image_a) at (0,0)
            {\includegraphics[width=0.5\columnwidth]{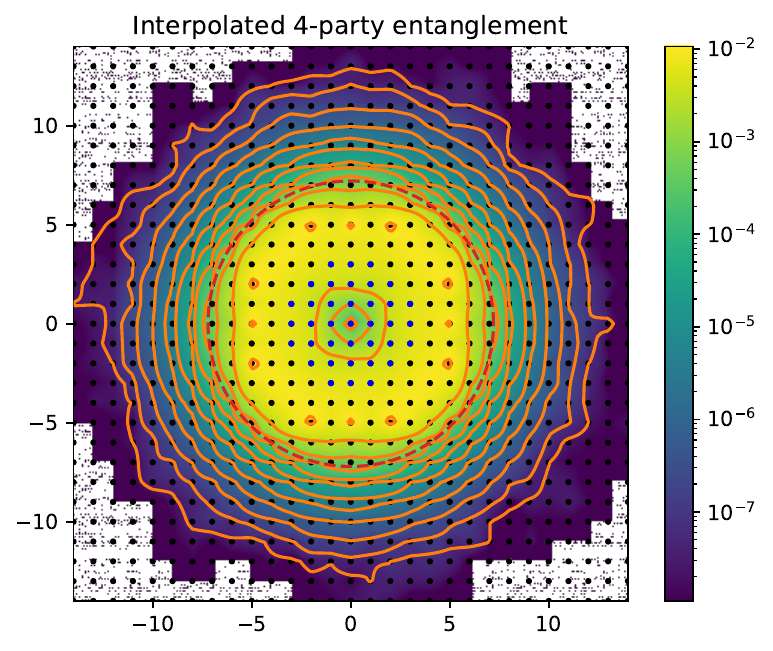}};
            \node [anchor=north west] (note) at (0,0) {\small{\textbf{d)}}};
        \end{scope}
    \end{tikzpicture}
    \vspace*{-0.4cm}
    \caption{\textit{a)} 2-party entanglement on the 2+1D measurement-only circuit with single-site subregions. Subregions are given in \textit{blue}, with the \textit{red dotted circle} indicating the distance beyond which neighboring subregions are not reduced in size due to overlaps with other subregions. Contours are given for every half an order of magnitude, while the entanglement has been extended beyond the lattice points \textit{(black)} using bilinear interpolation. The contours at intermediate distances are circular, implying that they are governed by a Euclidean distance metric. \textit{b-d)} Entanglement on the 2+1D measurment-only circuit with circular subregions of radius $\sqrt{13}$, for 2-4 parties. }\label{fig:top_down_entanglement}
\end{figure}

\subsection{Angle-averaged entanglement} \label{app:angle_averaged_entanglement}
As can be seen from Fig.~\ref{fig:top_down_entanglement}, the $k$-party entanglement $N_k(\eta)$ exhibits a rough circular symmetry, which gradually forgets the original shape of the subregion as the distance is increased. This symmetry allows us to define a specific ``average" entanglement at each $\eta$ value using our data.

Given a specific value of $\eta$ and angle $\theta \in [0,2\pi)$, we define $x(\eta, \theta)$ and $y(\eta, \theta)$ such that $\eta = \frac{c(2r)^2}{c(x)^2+c(y)^2}$ and $\tan \theta = \frac{c(y)}{c(x)}$. This specifies a coordinate in the original, non-periodic geometry, where we have determined $N_k(i,j)$ for every lattice point $(i,j)$. We can therefore extend this data to an estimate of $N_k(\eta, \theta)$ using bilinear interpolation. Finally, we average over $\theta$ to obtain $N_k(\eta)$.

The error in our estimation comes roughly from two sources. Firstly, there deviations in $N_k(\eta, \theta)$ over $\theta$, both from statistical fluctuations and systematic artifacts of the original shape of the subregions or lattice. Secondly, there is the shot noise error in our measured values $N_k(i,j)$, which will be reflected in the interpolated value of $N_k(\eta, \theta)$. While $N_k(\eta, \theta)$ is a weighted average over measurements at each lattice point:
\begin{gather}
    N_k(\eta, \theta) = \sum_{ij} w_{ij}(\eta, \theta) N_k(i,j) \qquad \qquad \sum_{ij} w_{ij}(\eta, \theta) = 1
\end{gather}
we make the conservative assumption that all lattice point measurements $N_k(i,j)$ have positively correlated fluctuations, contributing an error of 
\begin{gather}
    \sum_{ij} w_{ij}(\eta, \theta) \Delta N_k(i,j) = \sum_{ij} w_{ij}(\eta, \theta) \sqrt{\frac{N_k(i,j)}{\Omega}}
\end{gather}
where $\Omega$ is the number of iterations of the circuit.

As shot noise fluctuations form part of the deviations in $N_k(\eta, \theta)$ over $\theta$, these two errors share the same source to some extent, so assuming that they are instead independent should overestimate the overall error. We make this assumption, and add the errors from the two sources in quadrature. 

\subsection{Finite size scaling}\label{app:finite_size_scaling}
\begin{figure}[h]
    \centering
    \begin{tikzpicture}
        \begin{scope}
            \node[anchor=north west,inner sep=0] (image_a) at (0.0,0)
            {\includegraphics[width=0.84\columnwidth]{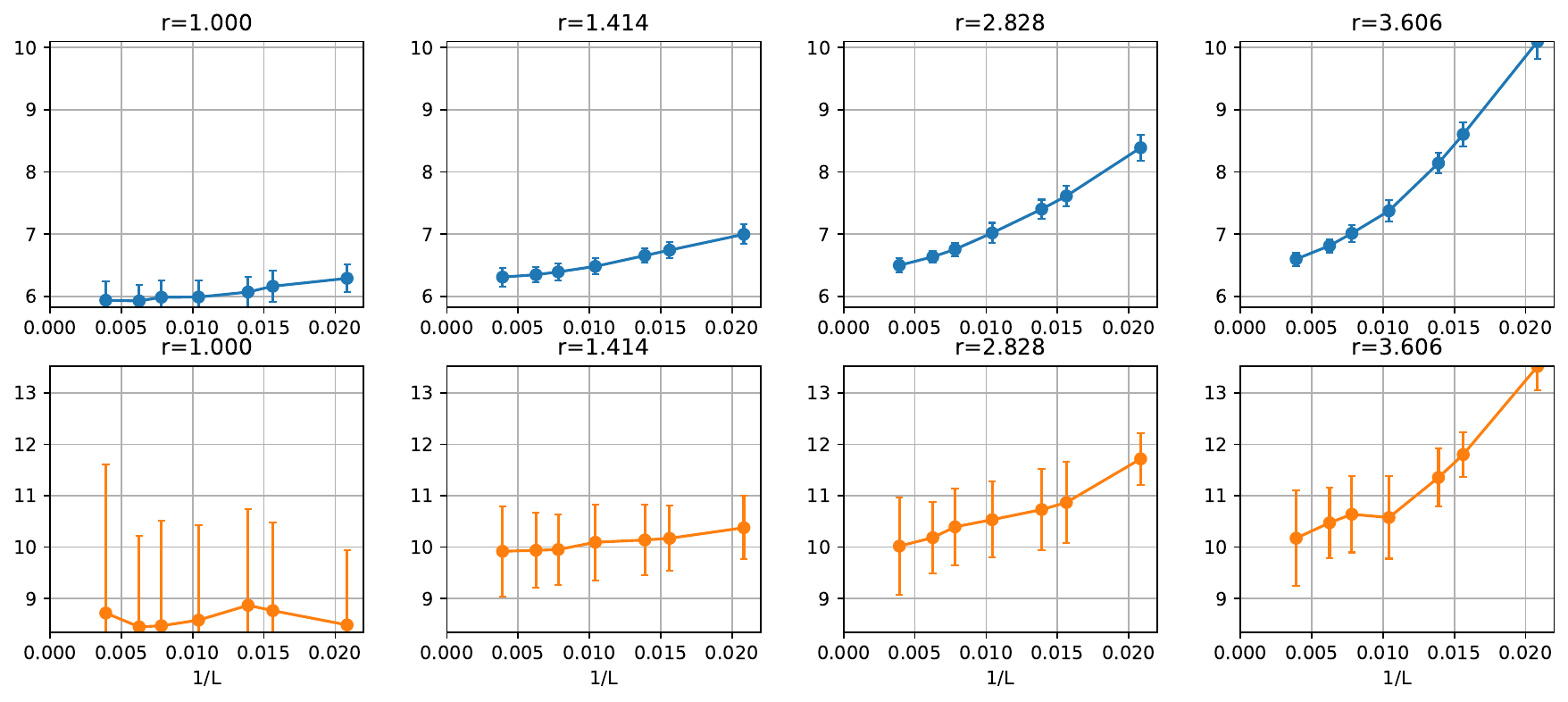}};
        \end{scope}
    \end{tikzpicture}
    \vspace*{-0.4cm}
    \caption{Fitted values of $\alpha_2$ and $\alpha_3$ for subregions of radius $1, \sqrt{2}$, $\sqrt{8}$ and $\sqrt{13}$ (corresponding to 1, 5, 21 and 37 lattice points respectively). }\label{fig:finite_size_scaling}
\end{figure}
The value of $\alpha_k$ varies over subregion radius (Fig.~\ref{fig:gme_3d}d). This is most likely due to finite-size effects, as the variation becomes less extreme when $L$ is increased (Fig.~\ref{fig:finite_size_scaling}). A small system size prevents large-radius subregions from reaching the limit of large relative separation: as can be seen from Fig.~\ref{fig:gme_3d}b,c, the apparent entanglement exponent tends to be steeper at shorter distance scales, biasing large-radius estimates upwards. Moreover, the effects of periodic boundary conditions become more apparent, which the chord length-based metric might not completely compensate for. Extrapolating the curves of Fig.~\ref{fig:finite_size_scaling} to the $1/L \rightarrow 0$ limit is difficult due to the ambiguous convexity of each curve, but it appears that the 2-party limit is somewhere between 5.9 and 6.5, while the 3-party limit is between 8.5 and 10.1.

\subsection{Indirect $k$-party entanglement}\label{app:indirect_entanglement}

If subregions $A_1 ... A_k$ are a single qubit each, they exhibit $k$-party entanglement in the MOC if and only if they belong to a $k$-qubit cat state. This condition gets relaxed if the subregions become larger than a single qubit, due to the freedom of local operations we can apply to each subregion. For example, suppose subregion $A_2$ consists of 2 qubits, one of which forms a Bell state with a site in $A_1$, while the other forms a Bell state with a site in $A_3$. Because a $ZZ$ measurement on those qubits in $A_2$ would yield a 3-party cat state over all three subregions, which can be converted into a 3-qubit GHZ once the $ZZ$ measurement outcome has been classically communicated to the other parties, we are forced to conclude that those subregions have 3-party GME, despite the lack of a 3-party cat state in the initial output. 

This indirect $k$-party entanglement complicates any measurement of $k$-party GME, but does not affect the entanglement exponents $\alpha_k$. This is because all cases of indirect entanglement must have harsher scaling than direct entanglement in the percolation model. In fact, any system with subadditive $\alpha_k$ should be determined by genuine network multipartite entanglement in the long distance limit, with indirect entanglement being a subleading or marginal term. The above example, where simultaneous $2+2$-party entanglement yielded indirect 3-party entanglement, would have a scaling exponent at least as bad as $2\alpha_2$ (and may be significantly \textit{worse} than that, due to the requirement of those events happening simultaneously without merging into direct entanglement - this is formalized in the BKR inequality~\cite{VanDenBerg1985,Reimer2000}).  The discrepancy between the scaling of direct and indirect entanglement for the 1+1D measurement-only circuit can be seen in Fig.~\ref{fig:indirect_entanglement}.

\begin{figure}[h]
    \centering
    \begin{tikzpicture}
        \begin{scope}
            \node[anchor=north west,inner sep=0] (image_a) at (0,0)
            {\includegraphics[width=0.6\columnwidth]{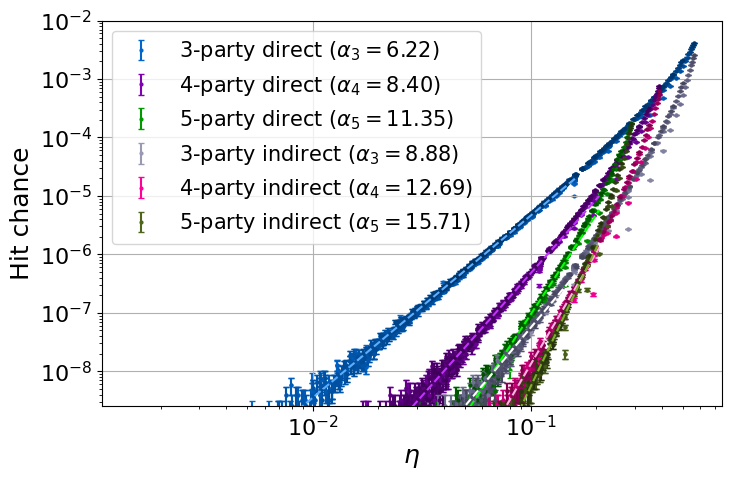}};
        \end{scope}
    \end{tikzpicture}
    \vspace*{-0.4cm}
    \caption{ Indirect entanglement on the 512 qubit 1+1D measurement-only circuit, compared to direct entanglement, for subregion widths 1-16. }\label{fig:indirect_entanglement}
\end{figure}

\section{Measurement-only stabilizer circuits with longer range entanglement}
\label{app:other_circuits}

In this section we will provide alternate examples of measurement-only circuits and ask whether the 1+1D alternating $ZZ/X$ architecture offers the longest-range entanglement possible. That is, we ask if it is possible to construct a random circuit ensemble on a system with some definition of locality and natural distance scale $x$, consisting of local two-site $ZZ$ and single-site $X$ measurements, such that $\alpha_k < 2k$. In addition, we require that the $k$-party entanglement between different subregions does not go to zero with system size - this rules out ensembles like the global Haar distribution where $\alpha_k$ is effectively $0$, but any local collection of sites are almost completely unentangled with each other.

There is an absolute lower bound on the decay rate of $k$-party entanglement in these circuits, derived from entanglement monogamy (which, with our class of stabilizer circuits, simply means that each site can only belong to one possible cat state). 
Suppose that the entanglement of $k$ sites $i_1 < ... <i_k$ has a lower bound
\begin{gather}
    p_0^{(k)} \prod_{j=1}^{k-1} |i_{j+1}-i_j|^{-\gamma},
\end{gather}
that is, $\alpha_k = \gamma(k-1)$. Then the probability that site $i_1$ forms the leftmost endpoint of \textit{any} $k$-site cat state is clearly lower bounded by
\begin{gather}
    p_0^{(k)}\zeta(\gamma)^{k-1}
\end{gather}
thus, in order for this probability to be bounded by 1, we must have $\gamma > 1$, with $p_0^{(k)}$ bounded by $\frac{1}{k\zeta(\gamma)^{k-1}}$ (Alternatively, specifying $\epsilon = \gamma-1 > 0$, we have $\alpha_k = (k-1)(1+\epsilon)$).

This bound is far more generous than the $\alpha_k = 2k$ scaling we found before, and in fact, there is very little stopping us from reaching this bound in a handful of layers on the right system. As long as the system geometry is wide enough to allow percolation clusters to cross each other (e.g. a ladder geometry of two parallel 1D systems), any product of cat states is the upper face of some percolation model. Therefore, for any distribution of products of cat states, we can define a random circuit ensemble that generates it, by directly relating each product state to a circuit that generates it. 

Of course, there is no guarantee that such an ensemble will exhibit any desirable qualities, such as a locally defined gate distribution or a connection to a CFT. In the rest of this section, we will present some ensembles that, while less artificial in construction than the example above, exceed the $\alpha_k = 2k$ scaling in different ways.

\subsection{Hyperbolic circuit}

This circuit acts on an $N = 2^n$-site system with periodic boundary conditions, using the same 1D geometry of the original circuit. Firstly, we``cut" the periodic BC at an arbitrary site chosen from the uniform distribution, setting that as our leftmost endpoint. This ensures that, while the circuit we create is not translation symmetric, the resulting ensemble is. 

The hyperbolic circuit consists of $n-1$ layers such that the $m$th layer from the top are connected by gates of size $2^m$. While these gates can be arbitrarily large, they can all be constructed by taking specific local $X$ and $ZZ$ measurements to finite depth.

At the core of each gate is a $ZZ$ measurement over all of its $2^m$ active sites. From there, we assign:
\begin{itemize}
    \item \textit{Left transmission}: A $\frac{p}{2}$ probability, with $p < 1$, of applying $X$ measurements to the left half of the output,
    \item \textit{Right transmission}: A $\frac{p}{2}$ probability of applying $X$ measurements to the right half of the output,
    \item \textit{Branching}: A $q < 1-p$ probability of applying no $X$ measurements at all,
    \item \textit{Reflection}: A $1-p-q$ probability of applying $X$ measurements to the entire input.
\end{itemize}

This system exhibits power law entanglement decay because, while the weight of a specific string is exponential over its length, the length of the string is effectively logarithmic in the distance it covers on the boundary, due to the hyperbolic geometry we are borrowing. Specifically,

\begin{claim}
    The probability that sites $A_1 < A_2 < ... < A_k$ form a generalized cat state is bounded below by
    \begin{gather}
        q^{k-2} (1-p-q) \frac{(p/2)^k}{2-(p/2)^k} x^{-k \log_2(\frac{2}{p})}
    \end{gather}
    where $x = |A_k - A_1|$.
\end{claim}
The first two terms of the bound comes from the necessity of $k-2$ branching gates and 1 reflection gate to form the path required to connect all $k$ sites, while the third term is a factor that comes from the translational invariance condition we applied to the ensemble. The final term comes from the requirement that all $k$ branches of the path reach their required destination - each branch consists of $O(\log(x))$ edges, as that is the approximate tree depth of the whole path, and each edge must select the correct transmission gate with probability $\frac{p}{2}$. Taking $p$ to be $1-\epsilon$ and $q, 1-p-q$ to be $\Theta(\epsilon)$, this probability bound roughly goes as $\epsilon^{k-1} x^{-k(1+\epsilon)}$. 

The exponent is only a lower bound, because there is still the probability that a trajectory gets stopped, completely and prematurely, on its way down by reflection nodes. While this probability is already accounted for whenever we require a trajectory to reach the end of a tree, it means that unwanted branching gates might not spoil the $k$-party entanglement if the unwanted trajectory gets stopped. At the large depth limit, this probability becomes $s=\min(1,\frac{1-p}{q}-1)$. This means that branching nodes might behave like transmission nodes instead if the unwanted branch ends prematurely - this effectively adds an extra $qs$ term to both transmission probabilities, shifting the exponent to $k\log_2(\frac{2}{p+2qs})$. This is assuming the branching node is always sufficiently deep to have stopping probability approximately equal to $s$, which might not be the case.

While this system has longer-range entanglement than the original MOC for all $k$, it has little connection to a CFT (at least in the original Euclidean geometry) due to the increasingly larger gates breaking the depthwise translation symmetry, and the sitewise translation symmetry only holding because of the decision to randomly move the entire circuit by some amount of sites. 

\subsection{Dyck word circuit}

Our next ensemble consists of only local, translation-symmetric gates, with even longer range 2-party entanglement than the hyperbolic circuit. At each layer, for each nearest-neighbor pair of sites $i,j$, we choose, with probability $p$, to measure $ZZ_{ij}$, then $X_i$ and $X_j$, then $ZZ_{ij}$ again. With probability $(1-p)$, we swap the two sites - though an equivalent to this gate can be constructed with measurements as well, provided we double the number of qubits at each site, or alternatively work on a ladder geometry. These circuits have similarities to a type of dual-unitary circuit with measurements, where each measurement forces the gates on its light cone to swaps~\cite{Claeys2022}. 

The resulting percolation model produces a distribution of string trajectories that travel around the bulk, and which are capable of changing direction vertically (from the $ZZ, X,X,ZZ$ measurement) but not horizontally. On the boundary, this produces a distribution of 2-site cat states, where sites $A$ and $B$ are part of the same cat state if they are connected by a trajectory string. For $p=\frac{1}{2}$, the weight of each possible string that connects $A$ and $B$ is $2^{-x}$, where $x$ is the distance between the two sites. The number of such strings is equal to the number of Dyck words of length $\frac{x+1}{2}$, which is the $\frac{x+1}{2}$\textit{th} Catalan number
\begin{align}
    C_{\frac{x+1}{2}} &\approx \frac{4^{\frac{x+1}{2}}}{\left(\frac{x+1}{2}\right)^{3/2}\sqrt{\pi}}\n 
    &\approx \frac{2^x}{x^{3/2}}
\end{align}
Therefore, the probability that $A$ and $B$ are entangled goes as $x^{-3/2}$ for large $x$. 
All values of $p$ produce the exact same scaling (Fig.~\ref{fig:dyck_word_circuit}a), with the only difference between the distributions being the constant prefactor, balanced by the probability of having a nearest-neighbor cat state. 

\begin{figure}[h]
    \centering
    \begin{tikzpicture}
        \begin{scope}
            \node[anchor=north west,inner sep=0] (image_a) at (0,0)
            {\includegraphics[width=0.4\columnwidth]{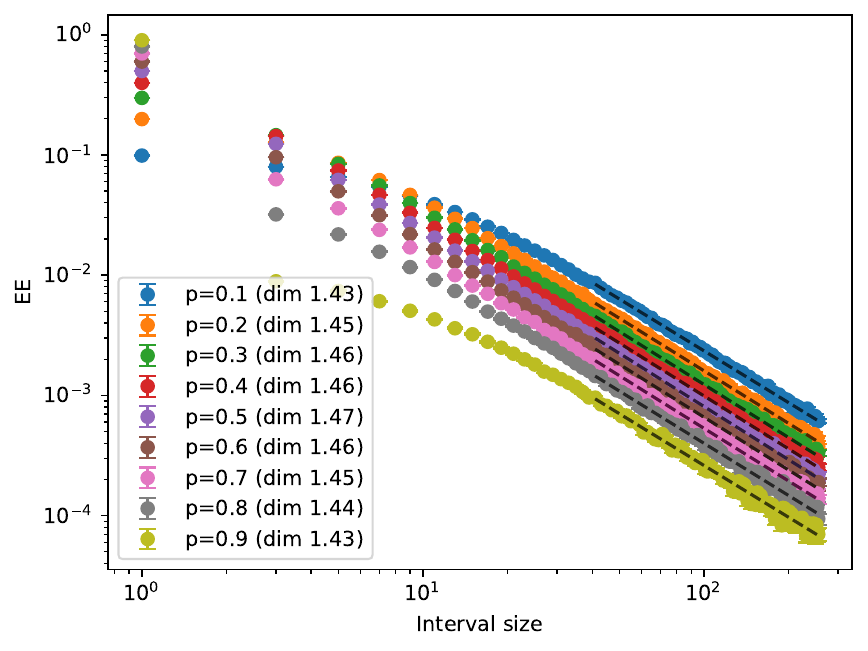}};
            \node [anchor=north west] (note) at (0,0) {\small{\textbf{a)}}};
        \end{scope}
        \begin{scope}[xshift=0.45\columnwidth]
            \node[anchor=north west,inner sep=0] (image_a) at (0,0)
            {\includegraphics[width=0.4\columnwidth]{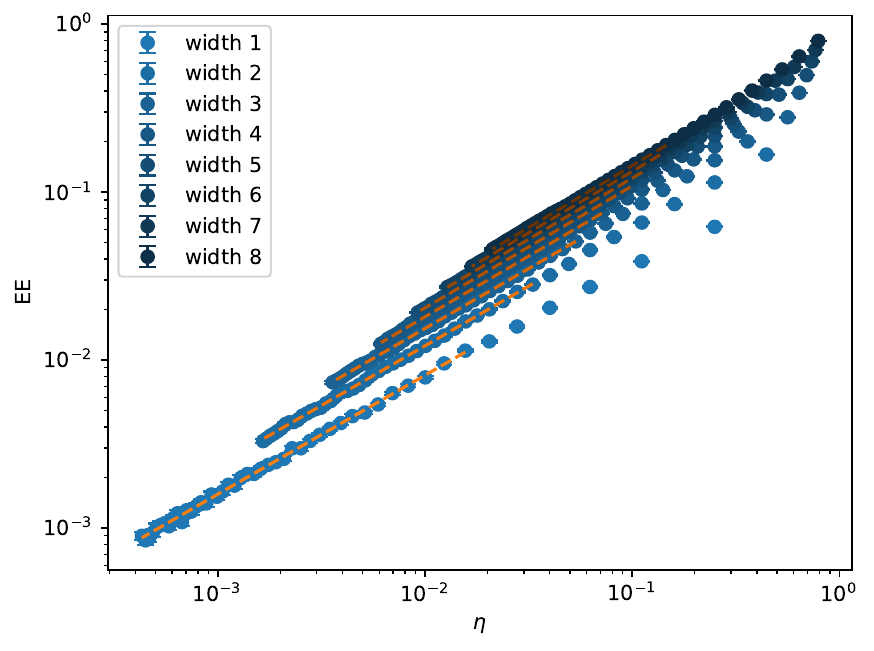}};
            \node [anchor=north west] (note) at (0,0) {\small{\textbf{b)}}};
        \end{scope}
    \end{tikzpicture}
    \vspace*{-0.4cm}
    \caption{(a) String chance in the Dyck word circuit over $p$ and distance. (b) Dyck word hit chance over $\eta$ for different subregion widths.}
    \label{fig:dyck_word_circuit}
\end{figure}

The entanglement distribution of this ensemble has no connection to a CFT, as there is no multipartite entanglement beyond $k=2$. Moreover, there is no asymptotic scale-invariance: the entanglement between different subregions is not solely a function of the anharmonic ratio $\eta$ even at large distance scales (Fig.~\ref{fig:dyck_word_circuit}b). Nonetheless, it is interesting that we can create an ensemble with such long-range entanglement, with robustness in at least one circuit parameter, just from a local, translation-invariant gate distribution.

\end{appendices}

\end{document}